\documentclass[runningheads]{llncs}
\newcommand{\longversion}[1]{#1}
\newcommand{\shortversion}[1]{}
\usepackage{todonotes}
\usepackage{comment}

\usepackage[utf8]{inputenc}          
\usepackage[T1]{fontenc}             
\usepackage[USenglish]{babel}          

\usepackage[intlimits]{amsmath}      
\usepackage{amsfonts}                
\usepackage{amssymb}                 
\usepackage{amsmath}                 

\shortversion{\usepackage{apxproof}
} 
\longversion{\usepackage[appendix=inline]{apxproof}}
\shortversion{\usepackage{apxproof}}%
\usepackage{mathtools}
\usepackage{setspace}
\usepackage{graphicx}
\usepackage{hyperref}
\usepackage{xcolor}
\usepackage{xspace}
\usepackage{amsthm}
\newtheorem{observation}{Observation}
\usepackage{caption}
\usepackage{subcaption}
\usepackage{tikz}
\usetikzlibrary{shapes.geometric}

\usepackage{algorithm}
\usepackage[noend]{algpseudocode}

\usepackage[most]{tcolorbox}
\newcommand{\problemdef}[4]{\setlength\tabcolsep{2pt}
	\begin{tcolorbox}[width = \columnwidth,colback=blue!5!white,colframe=blue!75!black,arc=0pt,outer arc=0pt,boxrule=1.5pt,left =0.5em,right=0em,enhanced,attach boxed title to top center={yshift=-3.8mm,yshifttext=-.9mm},colbacktitle=red!60,
  title=\textsc{#1} #2,
  boxed title style={size=small,colframe=red!50!black}]
  \mbox{}\\[-.5ex]
		\begin{tabular}{ @{\!\!\!}l p{0.89\columnwidth} c }
			\textsf{Input:} & #3 \\[.5pt]
			\textsf{Problem:} & #4
		\end{tabular}
	\vspace{-0.55em}
	\end{tcolorbox}
}

\newcommand{\NP}{\textsf{NP}}

\newcommand{\ptime}{\textsf{P}}

\newcommand{\iffl}{if\longversion{ and only i}f }
\newcommand{\ExtGMDefAll}{\longversion{\textsc{Extension Globally Minimal Defensive Alliance}}\shortversion{\textsc{ExtGMDefAll}}\xspace}
\newcommand{\ExtLMDefAll}{\longversion{\textsc{Extension Locally Minimal Defensive Alliance}}\shortversion{\textsc{ExtLMDefAll}}\xspace}
\newcommand{\no}{\textsf{no}\xspace}
\newcommand{\yes}{\textsf{yes}\xspace}

\newcommand{\nd}{\textbf{nd}}
\newcommand{\td}{\textbf{td}}
\newcommand{\tw}{\textbf{tw}}
\newcommand{\pw}{\textbf{pw}}
\newcommand{\vc}{\textbf{vc}}
\newcommand{\tc}{\textbf{tc}}

\newenvironment{pf}{\begin{proof}}{\hfill\qed\end{proof}}
\newenvironment{pfclaim}{\begin{proof}}{\hfill$\Diamond$\end{proof}}

\makeatletter
\newcommand{\crossout}[1]{%
  \begingroup
  \sbox\z@{#1}%
  \dimen\z@=\wd\z@
  \dimen\tw@=\ht\z@
  \dimen\z@=.99626\dimen\z@   
  \dimen\tw@=.99626\dimen\tw@ 
  \edef\co@wd{\strip@pt\dimen\z@}
  \edef\co@ht{\strip@pt\dimen\tw@}
  \leavevmode
  \rlap{\pdfliteral{q 1 J 0.4 w 0 0 m \co@wd\space \co@ht\space l S Q}}%
  \rlap{\pdfliteral{q 1 J 0.4 w 0 \co@ht\space m \co@wd\space 0 l S Q}}%
  #1%
  \endgroup
}
\makeatother

\begin{document}
\title{Enumerating Minimal Defensive Alliances}
\titlerunning{Enumerating Defensive Alliances}
\author{Zhidan Feng\inst1${}^,$\inst2\orcidID{0000-0002-3364-5396}  \and Henning Fernau\inst2\orcidID{0000-0002-4444-3220} \and Kevin Mann\inst2\orcidID{0000-0002-0880-2513}}

\authorrunning{Z. Feng, H. Fernau, K. Mann}
\institute{Shandong University, School of Mathematics and Statistic\\ 264209 Weihai, China.\\\and
Universit\"at Trier, Fachbereich~4 -- Abteilung Informatikwissenschaften\\  
54286 Trier, Germany.\\
\email{\{s4zhfeng,fernau,mann\}@uni-trier.de}
}

\maketitle

\begin{abstract}
In this paper, we study the task of enumerating (and counting) locally and globally minimal defensive alliances in graphs. We consider general graphs as well as special graph classes. From an input-sensitive perspective, our presented algorithms are mostly optimal.
\end{abstract}

\section{Introduction}

Alliance theory has been developed to model the formation of a group that can achieve a common goal if they cooperate. For instance, assuming that vertices represent social units (stretching from a single person to a whole nation, depending on the concrete scenario) of (for simplicity) equal power, while edges between vertices model vicinity (in some sense), a group of units may then be able to defend itself against `the rest' (and hence forms a defensive alliance) if each unit within the alliance has at least as much power when viewed together with its neighbors in the alliance as the power of the possibly united forces of the neighbors outside the alliance can bring into a fight.

This concept was originally introduced in \cite{FriLHHH2003,Kimetal2005,KriHedHed2004,SzaCza2001,Sha2004}. Up to now, hundreds of papers have been published on alliances in graphs\longversion{ and related notions}, as \longversion{also certified by}\shortversion{well as} surveys and book chapters\longversion{; see}  \cite{FerRod2014a,OuaSliTar2018,YerRod2017,HayHedHen2021}. 
\longversion{An overview on the wide variety of applications for alliances in graphs is given}\shortversion{Many applications for alliances in graphs are surveyed} in \cite{OuaSliTar2018}, including community-detection problems\longversion{ as described in} \cite{SebLagKhe2012}.
But when we want to find a defensive alliance in a graph, it is not very clear what further aims we like to achieve: For instance, do we like to find small or large alliances? Or what about bringing in costs, either associated to edges (modelling distances) or to vertices (possibly modelling the cost of maintaining an army there), or strengths (a vertex might possess more military abilities than another one), or other factors that are hard if not impossible to formalize?  In case of such unclear goals, it is at least a theoretical alternative to enumerate all alliances, potentially satisfying some further conditions like minimality (in some sense). Moreover, enumeration of solutions is quite a research area of its own, see~\cite{Was2016}, so that it is rather suprising that so far, to the best of our knowledge, no paper appeared on the enumeration of alliances. With this paper, we commence this natural line of research. We conclude this discussion by giving some references concerning enumeration of dominating sets (as this notion is quite related to alliances) and of hitting sets; this list of references is far from complete: \cite{AbuHeg2016,AbuFGLM2022a,BlaFLMS2022,BorGKM2000,BorEGK2002,GolHKS2020,KanLMN2014,KanUno2017,Mar2013a,Rei87\longversion{,Wot2001}}.
When we discuss enumeration algorithms for defensive alliances, we mostly focus on their input-sensitive analysis, i.e., we analyze their (mostly) exponential running time behavior and prove optimality by providing lower bound examples. When we try to obtain polynomial delay (one of the central concepts in the output-sensitive analysis of enumeration), we mostly focus on the technique of efficiently finding extensions of pre-solutions, very much in the sense of \cite{BorGKM2000,CasFGMS2022,Mar2013a}.

When we start thinking about enumerating or counting minimal alliances, we have to distinguish two notions of minimality introduced in the literature.
Shafique~\cite{Sha2004} called an alliance \emph{locally minimal} if the set obtained
by removing any vertex of the alliance is not an alliance. Bazgan \emph{et al.}~\cite{BazFerTuz2019}
considered another notion of alliance, called \emph{globally minimal}, 
with the property that no proper subset is an alliance. 
The non-monotonicity of alliances implies that there could be locally minimal alliances that are not globally minimal. This wording applies to any form of alliance, but we can clearly specialize it towards defensive alliances.
Of course, these two notions of minimality make no difference if we look for the smallest (non-empty) defensive alliance. This has been a rather classical line of research, see \cite{BliWol2018,FerRai07,GaiMaiTri2021,GaiMai2022,JamHedMcC2009,KiyOta2017}. Less (but recently increasing) work was done on finding largest minimal defensive alliances \cite{BazFerTuz2019,GaiMaiTri2021b,GaiMai2022a}.
Of course, algorithms for enumerating all minimal objects can also be used to find the smallest or largest among them. But enumeration is a task with quite unique features. In the discussion of this paragraph, let us only mention that a stronger form of monotonicity (namely, being hereditary) has been used for certain types of enumeration algorithms, see \cite{YuLLY2022} and the discussions therein.

\paragraph{Organization of the paper.}
The reader can find all necessary definitions, including some first observations, in \autoref{sec:defs}. In \autoref{sec:inputsense}, we discuss enumeration algorithms for minimal defensive alliances from an input-sensitive perspective. As we can show by providing lower-bound examples, even trivial enumeration algorithms are optimal in this sense.
Now, we switch our attention towards output-sensitive enumeration.
One strategy for obtaining enumeration algorithms with polynomial delay is to add into branching algorithms (typically used for enumeration algorithms that are good from an input-sensitive perspective) `fast' procedures to check if a continuation in a certain branch makes any sense. More formally speaking, this leads to extension problems as studied in \autoref{sec:extension-minDA}. Unfortunately, we will only give \NP-completeness results for many variations and graph classes. All this gives reasons for us to consider very simple graphs next, namely, trees. In \autoref{sec:tree-global-minDA}, we study the enumeration of globally minimal defensive alliances. Now, we can turn our (in terms of running times, measuring the input only) optimal enumeration algorithms into polynomial-delay algorithms, because the extension problem is polynomial-time solvable. This contrasts to what is known for locally minimal defensive alliances, where we can (only) provide some combinatorial (preliminary) results in \autoref{sec:tree-local-mindA} and have to leave (possibly surprisingly) the question whether or not all locally minimal defensive alliances of a given tree can be enumerated with an output-polynomial algorithm as an open question. We neither know the complexity status of the corresponding extension problem. In \autoref{sec:output-sensitive-minDA}, we return to the enumeration question concerning its globally minimal variation and provide another reason why we could not come up with an output-polynomial enumeration algorithm. Namely, we can show that, had we found such an algorithm, then we would have solved the \textsc{Hitting Set Transversal Problem}, a problem that is open now for more than four decades. Its solution would have considerable impact in different areas of Computer Science, such as Databases or Artificial Intelligence, as testified by its origins; see \cite{DomMisPit99,EitGot95,EitGotMak2003,GogPapSid98,KavPapSid93,ManRai87,Rei87}. Due to the particular non-monotonicity of the notion of defensive alliances, and also because there are potentially roughly $2^n$ many of them, it is interesting not only to list all minimal defensive alliances, but even the task of enumerating all defensive alliances becomes a task that is not completely trivial. We prove in \autoref{sec:enum-all-DA} that this task can be accomplished with polynomial delay, something which is unknown (and in a sense unlikely) for the minimal counterparts.
We list further points of future research in the concluding \autoref{sec:conclusions}.

\longversion{\section{Definitions}\label{sec:defs}}

In general, we use standard not(at)ion.
\shortversion{See its collection in the appendix.}
\longversion{Some of this is collected here for the convenience of the reader.}\begin{toappendix}
For a positive integer~$n$, let $[n]=\{ 1, \ldots, n\}$.

We assume some basic knowledge of graph theory and complexity theory on the side of the reader. Let us clarify some graph theoretic notions, though. To this end, let $G=(V,E)$ be an undirected graph. 
The \emph{depth} of a rooted tree is the maximum number of vertices on a path from root to leaf. A \emph{rooted forest} is a collection of rooted trees. The depth of a rooted forest is the maximum depth of the trees in the forest.  
$u,v\in V$ are called \emph{true twins} if $N[u]=N[v]$. Vertices of degree one are also called \emph{pendant}. 
$G$ is \emph{bipartite} if $V$ can be partitioned into $V_1$ and $V_2$ such that the induced graphs $G[V_1]$ and $G[V_2]$ contain no edges, i.e., if $V_1$ and $V_2$ are independent sets. Similarly, $G$ is a \emph{split graph} if $V$ can be partitioned into $V_1$ and $V_2$ such that $V_1$ is independent and $G[V_2]$ forms a complete graph, i.e., all possible edges are present. Graph~$G$ is a \emph{cograph} if every connected induced subgraph of~$G$ has diameter at most 2, or, equivalently, if $G$ has cliquewidth~2.

We now define several other important graph parameters. 
\begin{description}
    \item[$\nd(G)$] The \emph{neighborhood diversity} \cite{Lam2012} of a graph $G$, written $\nd(G)$, is the number of equivalency classes of the following equivalence relation: two vertices $u, v \in V$ are equivalent  if they have the same neighborhoods except for possibly
themselves, i.e., if $N(v) \setminus \{u\} = N(u) \setminus \{v\}$. The equivalence classes form either cliques or independent sets. More information on this parameter can be found in~\cite{Kou2013}. \longversion{Notice that $V$ can be partitioned into vertices of the same type in linear time using fast modular decomposition.}

    \item[$\tw(G)$] A \emph{tree decomposition} of~$G$ is given by a tree $T=(X,E_T)$ and a mapping $f:X\to 2^V$ such that (1) $\bigcup_{x\in X}f(x)=V$, (2) for each $\{u,v\}\in E$, there is some $x\in X$ with $f(x)\supseteq \{u,v\}$, and (3) for each $v\in V$, $X_v=\{x\in X\mid f(x)\ni v\}$ is connected in~$T$. $\tw(G,T,f)=\min\{\vert f(x)\vert - 1 \mid x\in X\}$ is the width of this decomposition, and the minimum width over all tree decompositions of~$G$ is the \emph{treewidth} of~$G$, written $\tw(G)$.
    \item[$\pw(G)$] If in the preceding definition, the tree happens to be a path, we speak of a path decomposition, and the minimum width over all path decompositions of~$G$ is the \emph{pathwidth} of~$G$, written $\pw(G)$.
    \item[$\td(G)$] An \emph{elimination forest} of $G$ is a rooted forest
$F$ with the same vertex set as $G$, such that for each edge $\{u,v\}$ of $G$, $u$ is an ancestor of $v$ or $v$ is an ancestor of $u$ in~$F$. (This forest can contain edges that are not in $G$.) The \emph{treedepth} of a graph $G$, denoted by $\td(G)$, is the minimum depth of an elimination forest of~$G$.
    \item[$\vc(G)$] A \emph{vertex cover} of a graph $G$ is a vertex set such that each edge is incident to at least one vertex in the set. $\vc(G)$ is the minimum size of a vertex cover of $G$.
    \item[$\tc(G)$] A set of vertices $X \subseteq V$ is a \emph{twin-cover} of $G$ (see~\cite{Gan2015}) if, for every edge $e = \{u,v\}\in E$, either (1) $u\in X$ or $v\in X$ [i.e., this part fulfills a vertex cover property], or $u,v$ are true twins, i.e., $N[u]=N[v]$. The size of the smallest twin-cover is also called, slightly abusing notation, the twin-cover (number) of~$G$, or $\tc(G)$ for short.
\end{description}
\end{toappendix}

Now, we formally define the central notions of our paper (as already described in the introduction). Let $G=(V,E)$ be an undirected graph. 
$A\subseteq V$ is a \emph{defensive alliance} if, for each $v\in A$, $\deg_A(v)+1\geq \deg_{\overline{A}}(v)$, where $\deg_X(v)$ is the number of neighbors $v$ has in the vertex set~$X$. A non-empty defensive alliance is called \emph{globally minimal} if no proper subset is a non-empty defensive alliance, while it is called  \emph{locally minimal} if no proper subset obtained by removing a single vertex is a non-empty defensive alliance. 
From the very definition, we can deduce the following fact that will be tacitly used in many places in this paper:
\begin{observation}
A globally minimal defensive alliance is connected.
\end{observation}

\noindent
For instance, this observation immediately implies: 
\begin{corollary}\label{cor:glob-def-all-paths}
There are only linearly many globally minimal defensive alliances in paths.
\end{corollary}

\begin{toappendix}
We can sharpen \autoref{cor:glob-def-all-paths} by saying that a $P_n$ has exactly $n-1$ many globally minimal defensive alliances. This can be seen by examining the proof of \autoref{lem:count_local_on_path}.
\end{toappendix}

\section{Input-sensitive Enumeration on General Graphs}
\label{sec:inputsense}

We first consider the enumeration problem from an input-sensitive perspective and show that even on very restricted graph classes, a seemingly trivial enumeration algorithm is optimal.

\begin{theorem}\label{thm:Enum_global_general}
There is an algorithm that enumerates all globally minimal defensive alliances in time $\mathcal{O}^*\left( 2^n\right)$. This algorithm is optimal, even for  complete bipartite graphs of pathwidth and treewidth~$2$, neighborhood diversity~$2$, vertex cover number~$2$, cliquewidth~$2$ (i.e., cographs), twin-cover~$2$, and tree depth~$3$.
\end{theorem}

 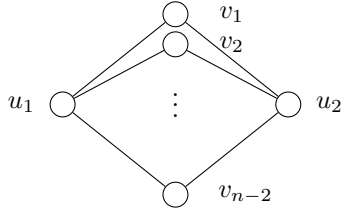
\begin{figure}[bt]
    \centering
    	
	\begin{tikzpicture}[transform shape]
		      \tikzset{every node/.style={ fill = white,circle,minimum size=0.3cm}}
			\node[draw,label={left:$u_1$}] (u1) at (-1.5,0) {};
			\node[draw,label={right:$u_2$}] (u2) at (1.5,0) {};
			\node[draw,label={right:$\ \,\,v_1$}] (v1) at (0,1.2) {};
			\node[draw,label={right:$\ \,\,v_2$}] (v2) at (0,.8) {};
			\node[draw,label={right:$\ \,\,v_{n-2}$}] (vn2) at (0,-1.2) {};
            \node at (0,0.1) {\vdots};					
            \path (u1) edge[-] (v1);
			\path (u1) edge[-] (v2);	
			\path (u1) edge[-] (vn2);					
            \path (u2) edge[-] (v1);
			\path (u2) edge[-] (v2);
			\path (u2) edge[-] (vn2);
        \end{tikzpicture}

    \caption{\label{fig:defall_diamond}A diamond $K_{2,n-2}$ constructed for \autoref{thm:Enum_global_general}.}
    \end{figure}

\begin{proof}
The naive enumeration algorithm is just going through all the cases by guessing for each vertex if it is in the solution or not, i.e., we are cycling through all vertex subsets. In each case, we check if the set is a globally minimal defensive alliance by using the polynomial-time algorithm from Proposition~5 of \cite{BazFerTuz2019}.  

Now, let us consider the complete bipartite graph $K_{2,n-2}\coloneqq (V,E)$ with $V \coloneqq \{u_1,u_2\} \cup \{v_1,\ldots,v_{n-2}\}$ and $ E \coloneqq \{\{u_1,v_i\},\{u_2,v_i\}\mid i\in [n-2]\}$ for some $n \in \mathbb{N}$ with $n>2$, which has a diamond shape as illustrated in \autoref{fig:defall_diamond}. A careful yet straightforward analysis shows the claimed structural properties\longversion{ of this example}.

\begin{toappendix} 
More precisely, we can say the following about $K_{2,n-2}=(V,E)$ with $V \coloneqq \{u_1,u_2\} \cup \{v_1,\ldots,v_{n-2}\}$ and $ E \coloneqq \{\{u_1,v_i\},\{u_2,v_i\}\mid i\in [n-2]\}$ for some $n \in \mathbb{N}$ with $n>2$. 
 Clearly, $\{u_1,u_2\}$ and $\{v_1,\ldots,v_{n-2}\}$ are the two neighborhood diversity classes. Hence, $\nd(K_{2,n-2})=2$. Also, $\{u_1,u_2\}$ forms a vertex cover, as well as a twin-cover, of $K_{2,n-2}$, so that $\vc(K_{2,n-2})=2$. Furthermore, $\{u_1,u_2,v_i\}$ for $i\in \{1,\ldots,n-2\}$ can be seen as the bags for a path decomposition of $K_{2,n-2}$, so that $\pw(K_{2,n-2})=2$. As forests are the graphs with treewidth (at most) one, also  $\tw(K_{2,n-2})=2$. The cograph property can also be easily checked. 
Finally, $\td(K_{2,n-2})=3$, because of the elimination tree $T=(V,E_T)$ with root $u_1$ and $E_T=\{\{u_2,w\} \mid w\in V\setminus \{u_2\}\}$. 
\end{toappendix}

Let $A' \subseteq \{v_1,\ldots,v_{n-2}\}$ with $\vert A' \vert = \lceil \frac{n-3}{2} \rceil$. Then $A \coloneqq \{ u_1 \} \cup A'$ is a globally minimal defensive alliance: 
For $v_i\in A$, $\deg_A(v_i) + 1 = 2 >\deg_{\overline{A}}(v_i)=1$. Further, $\deg_{A}(u_1) + 1 \geq \frac{n-3}{2} +1 = \frac{n-1}{2} = n-2-\frac{n-3}{2} \geq \deg_{\overline{A}}(u_1)$. Thus, $A$ is a defensive alliance. Assume there exists a nonempty defensive alliance $B \subsetneq A$. If $u_1\notin B$, then $\deg_B(v_i)+1=1<2= \deg_{\overline{B}}(v_i)$ for any $v_i\in B\subseteq A'$ would contradict the fact that $B$ is a defensive alliance. Hence, $B \cap \{v_1,\ldots,v_{n-2}\}\subsetneq A'$. So,  $\deg_{B}(u_1) + 1 < \frac{n-3}{2} +1 = \frac{n-1}{2} = n-2-\frac{n-3}{2} < \deg_{\overline{B}}(u_1)$. Therefore, $A$ is a globally minimal defensive alliance.

Since there are $\binom{n-2}{\lceil \frac{n-3}{2} \rceil}\geq \frac{1}{n-1}2^{(n-2) \cdot H\left(\frac{n-5}{2}\right)}\in o(2^n)$ many such sets (where $H$ is the binary entropy function), $\mathcal{O}(2^n)$ is a matching bound for the number of globally many defensive alliances on $K_{2,n-2}\coloneqq (V,E)$.\qed
\end{proof}

Since each globally minimal defensive alliance is also a locally minimal one, we directly get the next result. The algorithm mentioned is also simply going through all the cases. 

\begin{corollary}\label{cor:Enum_local_general}
There is an algorithm that enumerates all locally minimal defensive alliances in time $\mathcal{O}^*\left( 2^n\right)$. This algorithm is optimal, even for the restricted scenarios described in \autoref{thm:Enum_global_general}.
\end{corollary}

Later, we will add even one more scenario (treewidth~1) to the list when this trivial enumeration algorithm is optimal; see \autoref{lem:upperbound_local_caterpillar}.

\section{Extension Minimal Defensive Alliance}
\label{sec:extension-minDA}

Next, we move towards the output-sensitive perspective of enumeration.
We do this by looking at the extension problem that can be associated to our enumeration problem, defined formally below.
Hence, we cannot trivially re-write the branching algorithms from the previous section to achieve polynomial delay.
This type of results is also of interest independent of enumeration, as any branching algorithm that looks for the smallest non-empty alliance or for the largest minimal alliance could be potentially sped up by quickly deciding if certain branches are unnecessary to follow, as no minimal solutions are to be found in that subtree. More details and motivation for this line of research can be found in~\cite{CasFGMS2022}.

\problemdef{Extension Globally Minimal Defensive Alliance} {(\textsc{ExtGMDefAll})}{A graph $G=(V,E)$, and sets $U,N\subseteq V$.}{Is there a globally minimal defensive alliance $A\subseteq V\setminus N$ in $G$ with $U\subseteq A$?}

Quite similarly, we can also define \textsc{Extension Globally Minimal Defensive Alliance} (\textsc{ExtLMDefAll}).

To show \NP-hardness, we need \textsc{Cubic Planar Monotone 1-in-3-SAT} which  is \NP-complete according to \cite{MooRob2001}\longversion{, based on~\cite{Lar93}}. A Boolean formula is called \emph{monotone} either if each literal is positive or if each literal is negative. Without loss of generality, we only consider monotone Boolean formulae for which all literals are positive.  
For a monotone Boolean formula $\phi$ with the variable set $\mathcal{X}$ and the set of clauses $\mathcal{C}$, define its \emph{incidence graph} $G\left( \phi \right)=\left(V\left( \phi \right),E\left( \phi \right)\right)$ by setting $V\left( \phi \right)\coloneqq \mathcal{X} \cup \mathcal{C}$ and $E\left(\phi\right) \coloneqq \{ \{x,C\}\mid x\in C\}$.

\problemdef{Cubic Planar Monotone 1-in-3-SAT}{}{A monotone Boolean formula $\phi$ with a planar cubic incidence graph.}{Is $\phi$ satisfiable, such that each clause has exactly one true literal?}

\begin{figure}[bt]
    \centering
    	
	\begin{tikzpicture}[transform shape]
		      \tikzset{every node/.style={ fill = white,circle,minimum size=0.3cm}}
			\node[draw,label={left:$v_{i_1}$}] (x1) at (-1.25,1) {};
			\node[] (y11) at (-1.75,1.4) {};
			\node[] (y12) at (-1.5,1.6) {};
			\node[] (y13) at (-1.25,1.65) {};
			\node[] (y14) at (-1,1.6) {};
			\node[draw,label={left:$v_{i_2}$}] (x2) at (0,1.25) {};
			\node[fill = black,draw,label={above:$y_{i_2,1}$}] (y21) at (1.5,2.25) {};
			\node[fill = black,draw,label={above:$y_{i_2,2}$}] (y22) at (0.5,2.25) {};
			\node[fill = black,draw,label={above:$y_{i_2,3}$}] (y23) at (-0.5,2.25) {};
			\node[fill = black,draw,label={above:$y_{i_2,4}$}] (y24) at (-1.5,2.25) {};
			\node[draw,label={right:$v_{i_3}$}] (x3) at (1.25,1) {};
			\node[] (y31) at (1.75,1.4) {};
			\node[] (y32) at (1.5,1.6) {};
			\node[] (y33) at (1.25,1.65) {};
			\node[] (y34) at (1,1.6) {};
			\node[rectangle,fill = black,draw,label={below:$c_j$}] (cj) at (0,0) {};
			\node[rectangle,draw,label={below:$e_j$}] (ej1) at (-4,0) {};
			\node[fill = black,draw,label={above:$e_{j,1}$}] (ej11) at (-3,1) {};
			\node[fill = black,draw,label={above:$e_{j,2}$}] (ej12) at (-4,1.25) {};
			\node[fill = black,draw,label={above:$e_{j,3}$}] (ej13) at (-5,1) {};
			\node[rectangle,draw,label={below:$e_{j+1}$}] (ej2) at (4,0) {};
			\node[fill = black,draw,label={above:$e_{j+1,1}$}] (ej21) at (5,1) {};
			\node[fill = black,draw,label={above:$e_{j+1,2}$}] (ej22) at (4,1.25) {};
			\node[fill = black,draw,label={above:$e_{j+1,3}$}] (ej23) at (3,1) {};
			\node[fill = black,draw,label={right:$f_{j,1}$}] (fj1) at (2,-0.5) {};
			\node[fill = black,draw,label={right:$f_{j,2}$}] (fj2) at (2,-1.5) {};
			\node[rectangle,draw,label={below:$f_j$}] (fj) at (1,-1) {};
			\node[draw,isosceles triangle,isosceles triangle apex angle=60,minimum size =0.2cm,label={left:$z_{j,1}$}] (zj1) at (-2,-0.5) {};
			\node[draw,isosceles triangle,isosceles triangle apex angle=60,minimum size =0.2cm,label={left:$z_{j,2}$}] (zj2) at (-1.80,-1) {};
			\node[draw,isosceles triangle,isosceles triangle apex angle=60,minimum size =0.2cm,label={left:$z_{j,3}$}] (zj3) at (-1.5,-1.5) {};
			\node[] (leftend) at (-5,0) {$\ldots$};
			\node[] (rightend) at (5,0) {$\ldots$};
            \path (x1) edge[-] (cj);
            \path (x1) edge[-] (y11);
            \path (x1) edge[-] (y12);
            \path (x1) edge[-] (y13);
            \path (x1) edge[-] (y14);
            \path (x2) edge[-] (cj);
            \path (x2) edge[-] (y21);
            \path (x2) edge[-] (y22);
            \path (x2) edge[-] (y23);
            \path (x2) edge[-] (y24);
            \path (x3) edge[-] (cj);
            \path (x3) edge[-] (y31);
            \path (x3) edge[-] (y32);
            \path (x3) edge[-] (y33);
            \path (x3) edge[-] (y34);
            \path (ej1) edge[-] (cj);
            \path (ej2) edge[-] (cj);
            \path (zj1) edge[-] (cj);
            \path (zj2) edge[-] (cj);
            \path (zj3) edge[-] (cj);
            \path (fj) edge[-] (fj1);
            \path (fj) edge[-] (fj2);
            \path (fj) edge[-] (cj);
            \path (ej1) edge[-] (ej11);
            \path (ej1) edge[-] (ej12);
            \path (ej1) edge[-] (ej13);
            \path (ej2) edge[-] (ej21);
            \path (ej2) edge[-] (ej22);
            \path (ej2) edge[-] (ej23);
            \path (ej1) edge[-] (leftend);
            \path (ej2) edge[-] (rightend);
        \end{tikzpicture}

    \caption{\label{fig:construction_NP_ExtGMAD}Construction of $G$ of \autoref{thm:NP_ExtGMAD} for $j\in [m]$ ($j+1$ considered$\mod m$). The bipartiteness is explained by drawing black and white vertices. The vertex set~$U$ is illustrated by drawing vertices as squares.}
\end{figure}
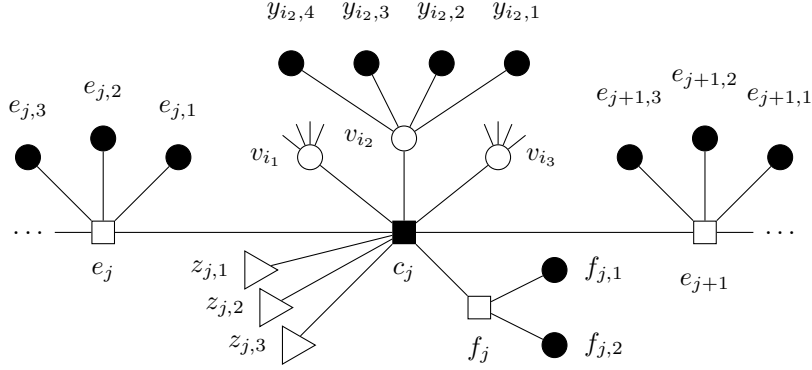
\begin{theorem}\label{thm:NP_ExtGMAD}
    \ExtGMDefAll is \NP-complete even on bipartite graphs of degree at most~9. 
\end{theorem}
\begin{pf}
    For \NP-membership, we\longversion{ can} employ guess-and-check, by\longversion{ using the algorithm of} Proposition~5 of~\cite{BazFerTuz2019}. 

    For the hardness result, we use \textsc{Cubic Planar Monotone 1-in-3-SAT}. Let $\phi$ be an instance with the variables $\mathcal{X}=\{x_1,\ldots,x_n\}$ and the clauses $ \mathcal{C} \coloneqq \{ C_1,\ldots, C_m\}$. Define $G\coloneqq (V,E)$ with 
    \begin{equation*}
        \begin{split}
            V \coloneqq& \{v_1,\ldots,v_n\} \cup \{ c_j \mid j \in [m]\} \cup \{e_j,e_{j,\ell}\mid j\in [m], \ell \in [3]\} \cup{} \\
            &\{ f_j, f_{j,t} \mid j \in [m],\, t \in [2] \} \cup \{ y_{i,k} \mid i \in [n],\, k \in [4] \} \cup  \{ z_{j,s} \mid j \in [m],\, s \in [3] \} \\
        \end{split}    
    \end{equation*}    \begin{equation*}
        \begin{split}
            E \coloneqq& \{\{ c_j, v_i\}\mid j \in [m],\, i\in [n],\, x_i \in C_j \} \cup \{ \{ y_{i,k}, v_i\}\mid  i \in [n],\, k \in [4] \} \cup{}\\
            &\{ \{ c_j, z_{j,s}\}\mid  j \in [m],\, s \in [3] \} \cup \{ \{ c_j, f_j\},\{ f_{j,t}, f_j\}\mid  j \in [m], t\in [2]\} \cup{}\\
            &\{ \{c_j,e_{j'}\},\{e_{j,\ell},e_{j} \} \mid j,j'\in [m],\, j' - j \equiv t \mod m, t \in \{0, 1\}, \ell \in [3] \}\,. 
        \end{split}    
    \end{equation*}
    Since each variable appears in three clauses, $\deg(v_i)=7$ for each $i \in [n]$. As $\phi$ is a \textsc{3-SAT} instance, $\deg(c_j)=9$ for each $j\in[m]$. The remaining vertices have degree at most~5. $G$ is bipartite, as indicated in \autoref{fig:construction_NP_ExtGMAD}, with the classes $\{ c_j, e_{j,\ell}, f_{j,t} \mid j \in [m],t\in [2], \ell \in [4]\} \cup \{ y_{i,k} \mid i \in [n],\, k \in [4] \}$ and $\{v_i\mid i\in [n]\} \cup \{ z_{j,s}, f_j, e_j \mid j \in [m], s \in [3]\}$.
    Define $U\coloneqq \{ c_j, e_j, f_j \mid j \in [m]\}$, $N\coloneqq \{ z_{j,s} \mid j \in [m], s\in [3]\}$\longversion{, illustrated by drawing vertices as squares and triangles, respectively, in \autoref{fig:construction_NP_ExtGMAD}}. 
     
    \begin{claim}
        $\phi$ is a \yes-instance of \textsc{Cubic Planar Monotone 1-in-3-SAT} \iffl there is a globally minimal defensive alliance $A\subseteq V \setminus N$ on $G$ with $U \subseteq A$.
    \end{claim}
    \begin{pfclaim}
    Let $a:\mathcal{X}\to \{0,1\}$ be a solution for the \textsc{Cubic Planar Monotone 1-in-3-SAT} instance $\phi$. Define $A\coloneqq U \cup \{v_i \mid x_i\in a^{-1}(1)\}$.
    Let $j\in [m]$. Then $\deg_A(e_j)+1=3=\deg_{\overline{A}}(e_j)$ and $\deg_A(f_j)+1=2=\deg_{\overline{A}}(f_j)$. Further, since each clause is satisfied by exactly one variable, $\deg_A(c_j) + 1 = 5 = \deg_{\overline{A}}(c_j)$. As further $\deg_A(v_i)+1=4=\deg_{\overline{A}}(v_i)$ for each $i\in [n]$, $A$ is a defensive alliance. Since each vertex fulfills the defensive alliance property with equality, if we delete a vertex from $A$ we have to delete also its neighbors from $A$ for $A$ to stay defensive alliance. Further, $A$ is connected ($\{ c_j, e_j \mid j \in [m]\} \subseteq A \subseteq N(\{ c_j, e_j \mid j \in [m]\} )$ and $\{ c_j, e_j \mid j \in [m]\}$ is a cycle). Thus, $A$ is globally minimal. 

    Let $A\subseteq V \setminus N$ be a globally minimal defensive alliance of $G$ with $U \subseteq A$. $A$ does not include any vertex from $\{e_{j,\ell}, f_{j,t}, z_{j,\ell} \mid j \in [m],\, \ell \in [3], t \in [2]\} \cup \{ y_{i,k}\mid i \in [n], k\in [4]\}$, as these are pendant and each would form a defensive alliance by itself. Hence, for each $j\in [m]$, there is an $i\in [n]$ such that $v_i\in A \cap N(c_j)$. Otherwise, $\deg_A(c_j) +1 = 4 < 6 = \deg_{\overline{A}}(c_j)$. Therefore, the assignment $$a: \mathcal{A}\to\{0,1\}, x_i\mapsto \begin{cases}
        1, & v_i \in A\\
        0, & v_i \notin A\\
    \end{cases}$$
    satisfies $\phi$. Assume there is a clause $C_j\in \mathcal{C}$ with two true literals. Define $A_j\coloneqq A\setminus \{f_j\}$. We want to show that $A_j$ is also a defensive alliance; this would contradict that $A$ is globally minimal and would prove that $\phi$ is a \yes-instance. As $\{ c_j \} = A_j\cap N(f_j)$, for $u\in A_j \setminus \{ c_j \}$, $\deg_{A_j}(u)=\deg_{A}(u) \geq \deg_{\overline{A}}(u)=\deg_{\overline{A_j}}(u)$. As $\deg_{A_j}(c_j) \geq 5\geq \deg_{\overline{A_j}}(c_j)$, $A_j$ is also a defensive alliance. 
    \end{pfclaim}
\noindent    Since there are $5 \cdot \vert \mathcal{X}\vert + 11 \cdot  \vert \mathcal{C}\vert$ many vertices, this reduction is polynomial. 
\end{pf}

We can tighten this theorem even more\longversion{ as follows} towards planar graphs\shortversion{ as explained in the appendix}.
Namely, the problem why the previous construction from \autoref{thm:NP_ExtGMAD} is not giving a planar graph is the cycle $\dots-e_j-c_j-\dots$ that is needed to ensure connectedness of the vertex set~$U$. But this problem can be overcome.

\begin{corollary}\label{cor:NP_ExtGMAD}
    \ExtGMDefAll is \NP-complete even on planar bipartite graphs of degree at most~9. 
\end{corollary}

\begin{toappendix}
\begin{figure}[htp]
    \centering
    \begin{tikzpicture}
        \tikzset{node/.style={fill = white,circle,minimum size=0.3cm}}

        \node[draw,circle] (x1) at (0,0) {};
        \node[draw,circle] (x2) at (-1,1) {};
        \node[draw,circle] (x3) at (0,1) {};
        \node[draw,circle] (x4) at (1,1) {};
         \node[draw,circle] (x5) at (-2,2) {}; \node[draw,circle] (x6) at (-1,2) {}; \node[draw,circle] (x7) at (0,2) {}; \node[draw,circle] (x8) at (1,2) {}; \node[draw,circle] (x9) at (2,2) {};
         \node[draw,circle] (x10) at (-1,3) {};
         \node[draw,circle] (x11) at (0,3) {};
         \node[draw,circle] (x12) at (0,4) {};
        
        \node[draw,rectangle] (c1) at (-0.5,0.5) {};
        \node[draw,rectangle] (c2) at (0.5,0.5) {};
        \node[draw,rectangle] (c3) at (-1.5,1.5) {};
        \node[draw,rectangle] (c4) at (-0.5,1.5) {};
        \node[draw,rectangle] (c5) at (0.5,1.5) {};
        \node[draw,rectangle] (c6) at (1.5,1.5) {};
        \node[draw,rectangle] (c7) at (-1.5,2.5) {};
        \node[draw,rectangle] (c8) at (-0.5,2.5) {};
        \node[draw,rectangle] (c9) at (-0.5,3.5) {};
        \node[fill = blue,draw,rectangle] (c10) at (3.5,2) {};
        \node[fill = blue,draw,rectangle] (c11) at (3.5,4.5) {};

\draw[red] (x1)--(x3)--(x7)--(x11)--(x12);
\draw[red] (x5)--(x6)--(x7)--(x8)--(x9);

        \path (x1) edge[-] (c1);
        \path (x1) edge[-] (c2);
        \path (x2) edge[-] (c1);
        \path (x3) edge[-] (c1);
        \path (x3) edge[-] (c2);
        \path (x4) edge[-] (c2);
        \path (x2) edge[-] (c3);
        \path (x3) edge[-] (c4);
        \path (x3) edge[-] (c5);
        \path (x4) edge[-] (c6);
        \path (x5) edge[-] (c3);
        \path (x5) edge[-] (c3);
        \path (x6) edge[-] (c3);
        \path (x6) edge[-] (c4);
        \path (x7) edge[-] (c4);
        \path (x7) edge[-] (c5);
        \path (x8) edge[-] (c5);
        \path (x8) edge[-] (c6);
        \path (x9) edge[-] (c6);
        \path (x5) edge[-] (c7);
        \path (x6) edge[-] (c7);
        \path (x6) edge[-] (c8);
        \path (x7) edge[-] (c8);
        \path (x10) edge[-] (c7);
        \path (x11) edge[-] (c8);
        \path (x10) edge[-] (c9);
        \path (x11) edge[-] (c9);
        \path (x12) edge[-] (c9);
        \path (x9) edge[-] (c10);
        \path (x12) edge[-] (c11);

        \draw[blue] (c1)--(c2);
        \draw[blue] (c1)--(-1.5,0.5)--(c3);
        \draw[blue] (c2)--(1.5,0.5)--(c6); 
        \draw[blue] (c5)--(c6);
        \draw[blue] (c4)--(c5);
        \draw[blue] (c3)--(c7);
        \draw[blue] (c9)--(3.5,3.5);
        \draw[blue] (c8)--(3.5,2.5);
        \draw[blue] (c4)--(c8);
        \draw[dotted,blue] (c7)--(-1.5,3.5)--(c9);

    \end{tikzpicture}
    \caption{Drawing a (blue) detour connection through all new clause vertices. Squares refer to clauses and circles refer to variables. The red lines indicate the original crossing. The blue squares denote the original clauses.}
    \label{fig:Laroche}
\end{figure}

\begin{pf}
We want to apply the same construction idea as for \autoref{thm:NP_ExtGMAD}, but we do not start with an instance of  \textsc{Cubic Planar Monotone 1-in-3-SAT}, but we first convert it into a planar layout by the construction described by Laroche~\cite{Lar93}. This means that we first lay out the formula on a grid, possibly destroying planarity. The advantage is that now, all clause vertices can be easily connected by a cycle (including intermediate vertices~$e_j$, as well as the decorations by the $z_*,f_*,e_*$-vertices that are attached like small trees), because the (original) clause variables all lie on a vertical line. Then, Laroche eliminates crossings of edges connecting variables and clauses one by one. However, this cross-over gadget, shown in \autoref{fig:Laroche}, admits a closed line drawn through all new clause vertices (as displayed in the figure). These lines (again including intermediate vertices etc.) will finally connect all clause vertices as required. More precisely, the figure shows the blue line of a detour concerning the rightmost crossing of a particular clause; the broken line indicates where the next detour might start to fix the second-to-rightmost crossing, etc. Moreover, notice that the newly introduced clauses all have three variables. Moreover, all new variable vertices have degree 2, 3, or 4. But even with variable vertices of degree~4, we can easily deal with by (now) adding five leaves $y_{i,1},\dots,y_{i,5}$ in our reduction, so that also in this slightly more involved construction, each vertex finally has degree at most nine.  \end{pf}
\end{toappendix}

Note that based on the construction of \autoref{fig:construction_NP_ExtGMAD}, any extension~$A$ of~$U$ with $A\cap N=\emptyset$ towards a locally minimal defensive alliance does not include any vertex from  $\{e_{j,\ell}, f_{j,t}, z_{j,\ell} \mid j \in [m],\, \ell \in [3], t \in [2]\} \cup \{ y_{i,k}\mid i \in [n], k\in [4]\}$, either. Since $U\subseteq A$,   $A\cap \{e_{j,\ell}, f_{j,t}\mid j \in [m],\, \ell \in [3], t \in [2]\}=\emptyset$; for each vertex in $ \{ y_{i,k}\mid i \in [n], k\in [4]\}$, if it is as an isolated vertex in $A$, then deleting it can get a smaller defensive alliance, if it is not as an isolated vertex in $A$, that is, its adjacent vertex $v_i$ is in $A$, as $\{c_j\mid\in [m]\}\subseteq U\subseteq A$, making $y_{i,k}$ redundant. Therefore, we have the following corollary:

\begin{corollary}\label{cor:NP_ExtLMAD}
    \ExtLMDefAll is \NP-complete even on planar bipartite graphs of degree at most~9. 
\end{corollary}


\section{Efficiently Enumerating Globally Minimal Defensive Alliances on Trees}\label{sec:tree-global-minDA}

In this section, we want to prove that there is a polynomial-delay enumeration algorithm for globally minimal defensive alliances which runs in time $\mathcal{O}^*\left(\sqrt{2}^n\right)$. We begin with showing a matching lower bound. Subsequently, we discuss some combinatorial lemmas for globally minimal defensive alliances on trees. Then, we prove that the extension problem associated to globally minimal defensive alliances can be efficiently solved. This is used for our enumeration algorithm that we present last, \longversion{which is }analyzed both from an output- and from an input-sensitive perspective.

\longversion{\subsection{A Lower-Bound Example Tree Family}}\shortversion{\paragraph{A Lower-Bound Example Tree Family}}

\begin{toappendix}
  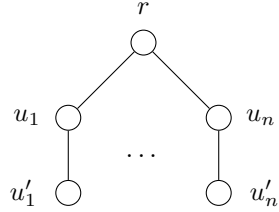
\begin{figure}[h]
    \centering
    	
	\begin{tikzpicture}[transform shape]
		      \tikzset{every node/.style={ fill = white,circle,minimum size=0.3cm}}
			\node[draw,label={above:$r$}] (r) at (0,0) {};
			\node[draw,label={left:$u_1$}] (u1) at (-1,-1) {};
			\node[draw,label={right:$u_n$}] (un) at (1,-1) {};
			\node[draw,label={left:$u'_1$}] (u'1) at (-1,-2) {};
			\node[draw,label={right:$u'_n$}] (u'n) at (1,-2) {};
            \node at (0,-1.5) {\ldots};			
            \path (r) edge[-] (u1);
			\path (r) edge[-] (un);		
            \path (u1) edge[-] (u'1);
			\path (un) edge[-] (u'n);
        \end{tikzpicture}

    \caption{\label{fig:defall_lowerbound_diameter-trees}Trees with low diameter: a construction for \autoref{lem:defall_lowerbound_diameter-trees}.}
    \end{figure}
\end{toappendix}

\begin{lemma}\label{lem:defall_lowerbound_diameter-trees}
There is a class of trees $T_n=(V_n,E_n)$ of order $n$ of diameter at most~4 with $\mathcal{O}^*\left(\sqrt{2}^n\right)$ many globally minimal defensive alliances. 
\end{lemma}
    
\begin{pf}
Let $n\in \mathbb{N}$. Define $T_n=(V_n,E_n)$ with $V_n=\{r\}\cup \{u_i,u'_i \mid i\in [n]\}$ and $E_n=\{\{r,u_i\},\{u_i,u'_i\}\mid i\in [n]\}\}$\longversion{; see \autoref{fig:defall_lowerbound_diameter-trees}}. Here $\vert V_n \vert = 2n+1$. Clearly, for any of the $n$ pendant vertices $u'_1,\ldots,u'_n$, this vertex is a defensive alliance by itself. Also, the distance between $u'_1$ and $u'_2$ is four, and any other pair of vertices has no greater distance from each other. So we only need to consider defensive alliances $A \subseteq \{r,u_1,\ldots,u_n\}$. If there is any vertex $u_i\in A$, then $r\in A$. Otherwise $\deg_A(u_i)+1=1<2= \deg_{\overline{A}}(u_i)$. So assume $r \in A$. Then at least $\lceil \frac{n-1}{2} \rceil$ many of its neighbors are in $A$. Otherwise, $\deg_{A}(r) + 1 < \frac{n-1}{2} +1 = \frac{n+1}{2} = n-\frac{n-1}{2} < \deg_{\overline{A}}(r)$. Conversely, let $A'\subseteq \{u_1,\ldots,u_n\}$ with $\vert A'\vert = \lceil \frac{n-1}{2} \rceil$, then $A' \cup \{ r \}$ is a defensive alliance: 
For $u_i\in A'$, $\deg_A(u_i) + 1 = 2 >\deg_{\overline{A}}(u_i)=1$ and $\deg_A(r) +1\geq \frac{n+1}{2} \geq \deg_{\overline{A}}(r)$. As we are looking for globally minimal defensive alliances, together with $r \in A$, exactly $\lceil \frac{n-1}{2} \rceil$ many of its neighbors are in $A$.

Hence, there are $n+\binom{n}{\lceil \frac{n-1}{2} \rceil}\in \mathcal{O}\left(2^n\right) = \mathcal{O}\left( \sqrt{2}^{\vert V \vert } \right)$ many globally minimal defensive alliances in $T_n$.   
\end{pf}

Observe that high-degree vertices are necessary in constructing trees with many globally minimal defensive alliances due to \autoref{cor:glob-def-all-paths}.

\longversion{\subsection{Combinatorial Results of Alliances in Trees}}\shortversion{\paragraph {Combinatorial Results of Alliances in Trees}}

\begin{lemma}\label{lem:_leaf_global_def_all}
    Let $T=(V,E)$ be a tree and $A \subseteq V$ be a globally minimal defensive alliance. Then \longversion{$}$\{ v \in A \mid \deg_A(v)=1 \} = \{ v\in V \mid \deg_V(v)\in \{ 2, 3\}\} \cap A\,.$\longversion{$}
\end{lemma}

\begin{pf}
Assume there exists a $v\in A$ with $\deg_V(v)\in \{2,3\}$ and $\deg_A(v)\neq 1$. If $\deg_A(v) = 0$, then $A$ is no defensive alliance. So assume $\deg_A(v) \in \{2, 3\}$. We show the result for $\deg_V(v) = 3$. For $\deg_V(v)=2$, the proof works analogously. 
    
Since $T$ is a tree, $v$ is a separator of at least three vertex sets, say, $X,Y,Z$.  As at least two neighbors of $v$ are in $A$, without lost of generality, $A_X \coloneqq X\cap A\neq \emptyset$ and $Y \cap A \neq \emptyset$. We now want to show that $A_X\cup\{v\}$ is already a defensive alliance. Since $v$ is a separator, $\deg_A(x)=\deg_{A_X}(x)$ for each $x\in A_X$. Therefore, we only need to consider $v$, for which $\deg_{A_X}(v)+1=2>1=\deg_{\overline{A_X}}(v)$ holds. Hence, $A_X\cup\{v\}\subsetneq A$ is a defensive alliance and $A$ not globally minimal. 

This leaves us to show $\{ v \in A \mid \deg_A(v)=1 \} \subseteq \{ v\in V \mid \deg_V(v)\in \{ 2, 3\}\}$. Each leaf~$x$ is a defensive alliance by itself and would hence satisfy $\deg_{A}(x)=0$ for any globally minimal defensive alliance~$A$ that contains~$x$. If there is a $v \in A$ with $\deg_V(v)\geq 4$ and $\deg_A(v)=1$, then $\deg_A(v)+1=2\leq 3 = \deg_{\overline{A}}(v)$ would be a contradiction.
\end{pf}

Let $T=(V,E)$ be a tree. Define $I(T)= \{v\in V \mid \deg_V(v)\in \{2,3\}\}$ and $L(T)=\{v\in V \mid \deg_V(v)=1\}$. 
Let $\mathcal{D}_g(T)$ be the set of  all globally minimal defensive alliances of~$T$.  

\begin{lemma}\label{lem:neighbored_I_vertices}
    Let $T=(V,E)$ a tree with $\{v,u\}\in E$ and $v,u\in I(T)$. Define $X_v$ as the connected component of $T_v=(V,E_v)$ with $E_v=E \setminus \{e\in E \mid u\in e \wedge v\notin e \}$ which includes $v$. Analogously define $X_u$. Then $$\mathcal{D}_g(T)=\left( \mathcal{D}_g(T[X_v]) \cup \mathcal{D}_g(T[X_u]) \cup \{\{ v,u\}\} \right) \setminus \{\{ v \}, \{ u \}\}\,.$$
\end{lemma}

\begin{pf}
Clearly, $\{v,u\}$ is a globally minimal defensive alliance of $T$ but not of $T[X_v]$ or $T[X_u]$. Conversely, $\{ v \}$ is a defensive alliance of $T[X_u]$ and $\{u\}$ one of $T[X_v]$ but none of these are defensive alliances of~$T$. 

Let $A \in \mathcal{D}_g(T[X_v]) \setminus \{\{ u \}\}$. Since $u\notin A$ and the remaining vertices have the same degree in $T[X_v]$ as in~$T$, $A$ is also a globally minimal defensive alliance of~$T$. Analogously, for each $A \in \mathcal{D}_g(T[X_u]) \setminus \{\{ v \}\}$, $A\in \mathcal{D}_g(T)$.

Let $A\in \mathcal{D}_g(T)$. If $A \cap X_v \neq \emptyset$ and $A \cap X_u \neq \emptyset$, then $v,u\in A$ and so $A =\{v,u\}$ by minimality. Hence, we now assume $A \subsetneq X_v \setminus \{u\}$ (or, symmetrically, $A \subsetneq X_u \setminus \{v\}$). Analogously to the arguments above, $A\in \mathcal{D}_g(T[X_u])$ (or $A\in \mathcal{D}_g(T[X_v])$, respectively).
\end{pf}

\begin{lemma}\label{lem:I_to_V_ratio}
    Let $T=(V,E)$ be a tree. If $\vert I(T) \vert \geq \frac{\vert V \vert}{2}$, then $\binom{I(T)}{2}\cap E \neq \emptyset$.
\end{lemma}

\begin{pf}
Assume $\vert I(T) \vert \geq \frac{\vert V \vert}{2}$ and $\binom{I(T)}{2}\cap E = \emptyset$.
Hence, each $v\in I(T)$ has a neighbor of degree at least~$4$, as otherwise $v$ would have a neighbor in $I(T)$ or $V=L(T)\cup\{ v \}$, i.e., $T$ would be a star $K_{1,2}$ or $K_{1,3}$, contradicting  $\vert I(T) \vert \geq \frac{\vert V \vert}{2}$ as $I(T)=\{v\}$.  If there is a $v \in I(T)$ with $\deg_V(v)=3$ and $N(v)=\{x,y,z\}$, where $\deg_V(x)\geq 4$, then define $T'\coloneqq (V',E')$ with $V'\coloneqq (V \setminus \{v\}) \cup \{a_v, b_v\}$ and $E'\coloneqq \left( E\setminus\{\{v,x\}, \{v,y\}, \{v,z\}\}\right)\cup \{\{x, a_v\}, \{x, b_v\}, \{y,a_v\}, \{z, b_v\}\}$. Clearly, the assumption also holds for $T'$ and $\vert I(T') \vert - \frac{\vert V' \vert}{2}\geq \vert I(T) \vert - \frac{\vert V \vert}{2}$. Therefore, we can further assume that there is no vertex of degree~$3$. Define $T_D=(V_D,E_D)$ with $V_D\coloneqq V' \setminus I(T')$ and $$E_D \coloneqq \left( E \setminus \{e\in E \mid e \cap I(T') \neq \emptyset \} \right) \cup \{ \{u,w\}\mid \exists v\in I(T'): \: \{u,w\} =N_{T'}(v) \}\,.$$ We can get $T'$ by subdividing the edges from~$T_D$. Therefore, $\vert I(T')\vert \leq \vert E_D\vert = \vert V'\vert- \vert I(T')\vert $. This is a contradiction to $\vert I(T') \vert \geq \frac{\vert V' \vert}{2}$.
\end{pf}

\longversion{\subsection{Extension Problem on Trees}}
\shortversion{\paragraph{Extension Problem on Trees}}

For our polynomial-delay result, we use the idea of extension algorithms.
\begin{theorem}\label{thm:ExtDA_onTree}
    \ExtGMDefAll is poly\-nomial-time solvable on trees. 
\end{theorem}
\begin{algorithm}[t]
\caption{Solving instances of \textsc{ExtGMDefAll}}\label{alg}
\begin{algorithmic}[1]
\Procedure{ExtGMDefAll Tree Solver}{$T,U,N$}\newline
 \textbf{Input:} A tree $T=\left(V,E\right)$ and two disjoint vertex sets $U,N\subseteq V$.\newline
 \textbf{Output:} Is there a globally minimal defensive alliance $A\subseteq V \setminus N$ \longversion{with}\shortversion{s.t.} $U \subseteq A$?
  \If{$U$ is not connected}
  \State Replace $U$ by the smallest connected vertex set containing~$U$.
  \EndIf
 \If{$U\cap N\neq\emptyset$}\longversion{\State} \textbf{return} \no.
 \EndIf
  \If{$U$ is a defensive alliance}
   \textbf{return} $\text{IsMinimalDefAll}(T,U)$.
   \State  (Test with the algorithm of Proposition 5 of \cite{BazFerTuz2019} on $(T,U)$)
  \EndIf
  \State Let $v \in U$ be a vertex with $\deg_U(v) + 1 < \deg_{\overline{U}}(v)$.
  \For{$u\in N(v)\setminus\left( U\cup N \right)$}
    \If{$\text{ExtGMDefAll Tree Solver}(T,U \cup \{ u \}, N)$} \textbf{return} \yes.
    \Else\ {Add $u$ to $N$.}
    \EndIf
\EndFor 
\State \textbf{return} \no.
\EndProcedure
\end{algorithmic}
\end{algorithm}
\begin{pf}     
Let $T=(V,E)$ be a tree and $U,N \subseteq V$. Since globally minimal defensive alliances are connected and between two vertices in a tree there is a unique path, put all vertices from paths between any two vertices in~$U$ into $U$. (By working bottom-up in an arbitrarily rooted tree, this can be implemented to run in linear time.) If $U\cap N \neq \emptyset$, this is a \no-instance. First check if $U$ is already a defensive alliance. Then, test if $U$ is a globally minimal defensive alliance: use the algorithm from Proposition~5 of \cite{BazFerTuz2019} and return a possible \yes-answer. 
    
Assume $U$ is no defensive alliance. Then there is a $v \in U$ with $\deg_U(v) + 1 < \deg_{\overline{U}}(v)$. If there is a $u \in N(v) \setminus (U \cup N)$ then run this algorithm with $U\cup \{u\}$ instead of~$U$. For the case that this recursive call returns \yes, also return \yes. Otherwise, put $u$ into $N$ and go on with the next vertex of  $N(v) \setminus (U \cup N)$. 

Since the algorithm returns \yes if a recursive call returns \yes or $U$ is a globally minimal defensive alliance and we only increase $U$, the \yes is correct. For the sake of contradiction, assume our algorithm returns \no even when there is a globally minimal defensive alliance $A \subseteq V \setminus N$ with $U \subseteq A$.  Choose $U$ such that $\vert U\vert$ is maximal with this property. Further, $\vert N\vert$ should be maximal. We can assume $U$ is not a defensive alliance, as $U\subseteq A$ and $U$ is no globally defensive alliance (otherwise, our algorithm would return \yes). Therefore, there is a $v\in U$ with $\deg_U(v)+1<\deg_{\overline{U}}(v)$. Let $u\in (A \cap N(v)\setminus (U \cup N)$. Now our algorithm would call itself also with $U\cup \{u\} \subseteq A$: because we assume that it returned \no it has to try all these neighbors. However, by the maximality of $U$, the algorithm should return \yes when given $(G,U\cup \{u\},N)$ as $U\cup \{u\} \subseteq A$. Therefore, our algorithm must also return \yes on  $(G,U,N)$, \longversion{which is an obvious contradiction to}\shortversion{contradicting} our assumption.

\shortversion{Clearly, the algorithm runs in polynomial time.}\longversion{Beside the recursive call, the algorithm only needs polynomial time. Furthermore, we call $\text{ExtGMDefAll Tree Solver}(T,U \cup \{ u \}, N)$ for each vertex~$u$ in~$V$ at most once. Therefore, the algorithm runs in polynomial time.}
\end{pf}

\longversion{\subsection{Enumeration Algorithm}}\shortversion{\paragraph{Enumeration Algorithm}}
\begin{theorem}
   All globally minimal defensive alliances on a tree of order $n$ can be enumerated in time $\mathcal{O}^*\left( \sqrt{2}^n\right)$ with polynomial delay.
\end{theorem}

\begin{pf}   
    Let $T=(V,E)$ be a tree. We use \autoref{lem:neighbored_I_vertices} extensively. By \autoref{lem:I_to_V_ratio}, we assume $I(T) < \frac{\vert V \vert}{2}$. First, we enumerate all defensive alliances with only one vertex. After that, we branch on each vertex in $I(T)$ if it is the solution or not. After each decision we run the algorithm of \autoref{thm:ExtDA_onTree}. After we made the decision for each vertex $I(T)$, we connect all vertices from $U$ by the unique path. 

    We do not enumerate any defensive alliance twice, as the defensive alliances differ in $I(T)$. This argument, together with the use of \autoref{thm:ExtDA_onTree}, implies that this algorithm runs with polynomial delay.

    This leaves to show that we enumerate each globally minimal defensive alliance. So let $A \in \mathcal{D}_g(T)$ with $\vert A \vert > 1$. Hence, $A$ is connected and $T[A]$ is a subtree. By \autoref{lem:_leaf_global_def_all}, $\vert I(T) \cap A \vert \leq 2$. Let $L' = I(T) \cap A$.  By our branching, we should come into this situation, otherwise this would contradict \autoref{thm:ExtDA_onTree}. Let $A'$ be the minimal globally defensive alliances which is output in this case. Since $A$ is connected, $A' \subseteq A$. If $A' \subsetneq A$, then $A'$ is no defensive alliance. Assume $A'$ is no defensive alliance. Then there is a $v\in A'$ with $\deg_{A'}(v) + 1 < \deg_{\overline{A'}}(v)$. Therefore, there exists a $u \in (N(v) \cap A)\setminus A'$. This implies that there is a leaf of $T[A]$ below $u$ which is in $A\setminus A'$, contradicting the construction of~$A'$.  
\end{pf}

We want to stress once more that due to our lower-bound example, the previous algorithm cannot be substantially improved.

\section{Listing Locally Minimal Defensive Alliances on Trees}
\label{sec:tree-local-mindA}

We first show that there are very simple families of trees (namely, trees that are both caterpillars and subdivided stars) that actually host (nearly) as many locally minimal defensive alliances as described by the upper bound for general graphs given in \autoref{cor:Enum_local_general}. Then, we prove that at least for paths, less locally minimal defensive alliances can be found, but still exponentially many. Notice that this contrasts our findings in the global case, see \autoref{cor:glob-def-all-paths}.

\begin{toappendix}
    \begin{figure}[bt]
    \centering
    	
	\begin{tikzpicture}[transform shape]
		      \tikzset{every node/.style={ fill = white,circle,minimum size=0.3cm}}
			\node[draw,label={left:$a$}] (a) at (0,0) {};
			\node[draw,label={below:$b$}] (b) at (1,0) {};
			\node[draw,label={below:$c$}] (c) at (2,0) {};
			\node[draw,label={left:$v_1$}] (v1) at (-0.75,-1.5) {};
			\node[draw,label={right:$v_k$}] (vk) at (0.75,-1.5) {};
            \node at (0,-1.5) {\ldots};			
            \path (a) edge[-] (b);
			\path (b) edge[-] (c);		
            \path (a) edge[-] (v1);
			\path (a) edge[-] (vk);
        \end{tikzpicture}

    \caption{\label{fig:defall_caterpillar}Caterpillar construction for \autoref{lem:upperbound_local_caterpillar}.}
\end{figure}
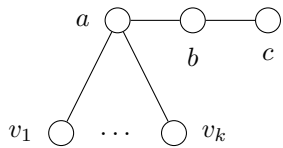
    
\end{toappendix}

\begin{lemma}\label{lem:upperbound_local_caterpillar}
    There is a class of caterpillars (\longversion{that can be also interpreted as}\shortversion{or of} subdivided stars) $T_n=(V_n,E_n)$ of order $n$  with $\Theta\left(2^n\right)$
    many locally minimal defensive alliances. 
\end{lemma}

\begin{pf}
By the algorithm of \autoref{cor:Enum_local_general}, we know there are at most $2^n$ locally minimal defensive alliances on a graph of order~$n$. 

We now show that this trivial bound can be reached. To this end, define $T_{k+3}=T=(V,E)$ with $k \in \mathbb{N}$, $V=\{a,b,c,v_1,\ldots,v_k\}$ and $E = \{\{a,b\}, \{b,c\}\} \cup \{\{a,v_i\} \mid i \in [k]\}$.\longversion{ This construction is illustrated in \autoref{fig:defall_caterpillar}.}

Let $A' \subseteq \{v_1,\ldots,v_k\}$ with $\vert A' \vert =\lceil \frac{k-2}{2}\rceil$. Define $A=\{ a,b\} \cup A'$. Since the vertices $v_i \in A'$ have degree~1, $\deg_A(v_i)\geq \deg_{\overline{A}}(v_i)$. Further, $\deg_{A}(b) + 1 =2 > 1 = \deg_{\overline{A}}(b)$ and $\deg_{A}(a) + 1 \geq \frac{k-2}{2} + 1 + 1 =  \frac{k+2}{2} = k - \frac{k - 2}{2} \geq  \deg_{\overline{A}}(a)$. Hence, $A$ is a defensive alliance. 

$A\setminus \{a\}$ is no defensive alliance, since $\deg_{A\setminus \{a\}}(b)+1=1<2=\deg_{\overline{A\setminus \{a\}}}(b)$. For $v\in A\setminus \{a\}$,  $A\setminus \{v\}$ is also no defensive alliance, as $\deg_{A\setminus \{v\}}(a)+1=\lceil \frac{k-2}{2}\rceil+1 \leq \frac{k-1}{2} + 1< \frac{k+2}{2}\leq\deg_{\overline{A\setminus \{v\}}}(a)$.

As also leaves by themselves form defensive alliances, in $T_{n}$ (of order~$n$), one finds
$(n-3)+\binom{n-3}{\lceil \frac{n-5}{2}\rceil}\approx \gamma\cdot 2^n$ many  locally minimal defensive alliances, for some constant~$\gamma$.
\end{pf}

It is an open question if all locally minimal defensive alliances can be enumerated with polynomial delay.

We now turn our attention to paths, a\longversion{n admittedly} quite restricted class of trees. Finally, here we find fewer  locally minimal defensive alliances than in the general case.

\begin{lemma}\label{lem:count_local_on_path}
On a path of order $n$, there are $f_{n-2} + 2 \in \mathcal{O}(1.4656^{n})$ many locally minimal defensive alliances, with $f_1=0$, $f_2=1$, $f_3=2$ and, for $n\geq 4$, $f_n = f_{n-1} + f_{n-3} + 1$ .
\end{lemma}

\begin{pf}
Let $P_{n}\coloneqq( V_{n}, E_n)$ with $V_n \coloneqq \{v_1,\ldots,v_n\}$, $E_n \coloneqq \{\{ v_i,v_{i+1}\} \mid i\in [n]\}$ and $A \subseteq V_n$ being a locally minimal defensive alliance. 
    
First, we show that any connected component of  $P_n[A]$ has at most 2 vertices. We prove this by contradiction. Assume that there is an $i\in [n-2]$ with $v_i,v_{i+1},v_{i+2}\in A$. Let $i$ be the largest such integer. This implies that $v_{i+3}$ either does not exist or $v_{i+3} \notin A$. In both cases, $v_{i+1}$ is the only neighbor of $v_{i+2}$ in~$A$. Hence for each $v_j\in A\setminus \{v_{i+1}, v_{i+2}\}$, $\deg_{A \setminus \{v_{i+2}\}}(v_j) + 1 = \deg_{A}(v_j) + 1 \geq \deg_{\overline{A} }(v_j)=\deg_{\overline{A \setminus \{v_{i+2}\}}}(v_j)$. Furthermore, $\deg_{A \setminus \{v_{i+2}\}}(v_{i+1}) + 1 = 2 > 1 =\deg_{\overline{A \setminus \{v_{i+2}\}}}(v_{i+1})$. Hence, $A \setminus \{v_{i+2}\}$ is a defensive alliance and $A$ is not minimal. Therefore, we can assume that the connected components of $P_n[A]$ have at most 2 vertices. In fact, there are only two situations when a connected component of $P_n[A]$ has one vertex, as we show next.
    
Next, we claim that, if $v_1 \in A$, then $\{v_1\} = A$. For the sake of contradiction, assume $v_1,v_2 \in A$. Then, $v_3\notin A$ and $v_1$ is the only neighbor of $v_2$ in~$A$. Hence, $A \setminus \{v_2\}$ is also a defensive alliance. Therefore,  $v_1\in A$ has no neighbors in $A$. Hence, if $\{v_1\}\subsetneq A$, then $A \setminus \{v_1\}$ is also a defensive alliance. 
    
Analogously, if $v_n \in A$ then $\{v_n\} = A$. 
    
Furthermore, there is no other vertex $v_j\in A$ which is a connected component by itself, as otherwise $ \deg_{A }(v_j) + 1 <2= \deg_{\overline{A} }(v_j)$. 

For $n\in \mathbb{N}$, define the set of alliances $\mathcal{A}_n$ as $$\{ A\subseteq V_{n+2}\mid A\text{ is a locally minimal defensive alliance of }P_{n+2}:v_1,v_{n+2}\notin A  \}$$ and  $f_n \coloneqq \vert \mathcal{A}_n\vert$. Clearly, the initial conditions are correct.
\begin{toappendix}
More precisely, we can argue for the values of $f_1,f_2,f_3$ as follows:
\begin{itemize}
    \item $f_1=0$ ($\{v_2\}$ is not a defensive alliance of $P_3$),
    \item $f_2=1$ ($\{v_2, v_3\}$ is the only defensive alliance of $P_4$ fulfilling the property) and
    \item $f_3=2$ ($\{v_2, v_3\},\{v_3, v_4\}$ are the only defensive alliances of $P_5$ fulfilling the property).
\end{itemize}
\end{toappendix}
So assume $n\in \mathbb{N}$ with $n > 3$. For $v_2$, consider two possibilities: $v_2 \notin A$ or $v_2\in A$.

\smallskip\noindent
\textbf{Case 1:} $v_2 \notin A$. There is a bijection $B_{\text{out}}: \mathcal{A}_{n-1} \to  \mathcal{A}_n \cap 2^{V_{n+2} \setminus \{v_2\}}$ by $B_{\text{out}}(A)=\{v_{j+1} \mid v_j\in A\}$. The inverse function is given by $B_{\text{out}}^{-1}(A)=\{v_{j-1} \mid v_j\in A\}$. These functions are well-defined, as for each $A \in \mathcal{A}_{n-1}$ and $v_j\in A$, $\deg_{A}(v_j)= \deg_{B_{\text{out}}(A)}(v_{j+1})$ and $\deg_{\overline{A}}(v_j)= \deg_{\overline{B_{\text{out}}(A)}}(v_{j+1})$. Hence, the number of locally minimal defensive alliance $A\subseteq V_{n+2}$ on $P_{n+2}$ with $v_1,v_2,v_{n+2}\notin A$ is~$f_{n-1}$.       

\smallskip\noindent
\textbf{Case 2:} $v_2 \in A$. Since $v_1\notin A$ by definition of $ \mathcal{A}_n$, $v_3\in A$ and $v_4\notin A$ by  the observations from above. Then there are two subcases: $A=\{v_2,v_3\}$ and $\{v_2,v_3\}\subsetneq A$. We only need to consider the second subcase more thoroughly. Define $\mathcal{A}'_n \coloneqq \{A \in \mathcal{A}_n\mid v_2,v_3\in A \wedge A\setminus \{v_2,v_3\}\neq \emptyset\}$. Then there is a bijection $B_{\text{in}}: \mathcal{A}_{n-3} \to \mathcal{A}'_n$. This is a well-defined bijection working analogously to Case~1. In total, $f_n=f_{n-1}+f_{n-3}+1$, where the `$+1$' covers the case $A=\{v_2,v_3\}$. 

    For $n\in \mathbb{N}$, define $g_n=f_n+1$. Hence, $g_n=f_n + 1= f_{n-1}+f_{n-3}+1 + 1= g_{n-1} + g_{n-3}$. Thus, $f_n,g_n\in \mathcal{O}(1.4656^n)$.
\end{pf}

Also, the proof of the preceding combinatorial result also provides an algorithm for enumerating all locally minimal defensive alliances of a path, and this algorithm works with polynomial delay.

\begin{toappendix}

\section{Output-sensitive Enumeration of Defensive Alliances}
\label{sec:output-sensitive-minDA}

Our previous studies on extension variations of defensive alliance problems ruled out output-polynomial time algorithms that are based on naive branching algorithms. But there are other strategies to obtain output-polynomial time or even polyno\-mial-delay enumeration algorithms, for instance, one could define a distance function on the space of all minimal solutions and hope to `walk' through this space without erring too much when walking from one solution to another (new) one, see \cite{ConGMUV2022,KobKurWas2022,KobKurWas2022a,KobKMO2024,YuLLY2022} for recent references, or one could try to define an easily computable bijection between two solution spaces where one is known to be enumerable in output-polynomial time, see~\cite{ManFer2024} for a recent example.

One possibility to completely rule out any of the mentioned attempts would be to show a kind of reduction that tells us that if we could have such an enumeration algorithm for listing minimal alliances, we would be also in the position to list solutions in a space where this is known to be impossible.
Alas, to the best of our knowledge, no examples of such spaces are known. 
However, there is one very famous example that so far resisted all attempts to efficiently enumerate all solutions, also known as the \textsc{Hitting Set  Transversal Problem}. Equivalently stated, the question is if all minimal dominating sets of a graph can be enumerated in output-polynomial time.\footnote{Notice that for dominating sets, one need not distinguish between locally and globally minimal solutions due to the monotonicity of this property: any superset of a dominating set is dominating.} This question is open for more than four decades. It is quite important, as it appears in may application areas, and in particular in databases, there are quite a number of interesting equivalent problems, or problems that are shown \emph{transversal-hard}, as called in \cite{GogPapSid98}; 
we only mention two recent references and refer to the papers cited therein: \cite{BlaFLMS2022,BlaFriSch2022}.
The following theorems imply that enumerating globally minimal defensive alliances are at least as hard as enumerating minimal dominating sets.
To some extent, this result therefore tells us that not only our attempt to use efficient algorithms for solving some extension variant failed, based on $\NP\neq\ptime$, but also any other strategy is likely to fail, as this would solve an old open problem.

We will actually prove transversal hardness for two restricted scenarios: enumerating globally minimal
defensive alliances on bipartite graphs and on split graphs.  This is interesting insofar, as listing minimal dominating sets is also transversal-hard on bipartite graphs, but not on split graphs \cite{KanLMN2014}.

\begin{theorem}\label{thm:outpolyhard_bipartite}
If there would be an output-polynomial time algorithm for enumerating globally minimal
defensive alliances on bipartite graphs, then there would be an output-polynomial time algorithm for enumerating minimal dominating sets in general connected graphs. 
\end{theorem}

\begin{proof}
    Let $G=(V,E)$ be a connected graph. Define $V_i := \{ v_i\mid v \in V\}$ for $i \in [4]$, $V'_1:= \{v_{1,i}\mid v\in V, i\in [\deg(v) + 2]\}$, $V'_2:= \{v_{2,1},v_{2,2}\}$, and $V'_3:= \{v_{3,i}\mid v\in V, i\in [2\deg(v) + 3]\}$, $V'_4:= \{v_{4,i}\mid v\in V, i\in [\deg(v) + 2]\}$ as well as $G'=(V',E')$ with 
    \begin{equation*}
        \begin{split}
            V' \coloneqq{}& V_1 \cup V_2 \cup V_3 \cup V_4 \cup V_1' \cup V_2' \cup V_3' \cup V_4'\,,\\
            E' \coloneqq{}& \{\{v_i,v_{i,j}\}\mid i\in [4], v\in V, v_{i,j}\in V_i'\} \cup {}\\
            &\{\{v_2, u_1\},\{v_2,u_3\},\{v_4,u_3\} \mid v\in V, u\in N[v]\}.
        \end{split}
    \end{equation*}
     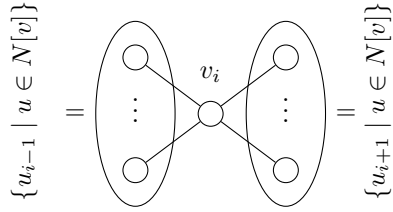
\begin{figure}[bt]
    \centering
    	
	\begin{tikzpicture}[transform shape]
		      \tikzset{every node/.style={ fill = white,circle,minimum size=0.3cm}}
            \node  at (3.7,0)[label={[rotate=90]above:$\{u_{i+1}\mid u\in N[v]\}$}]{};
            \node [label={above:$=$}] at (1.8,-0.5){};
            \node  at (-1,0)[label={[rotate=90]above:$\{u_{i-1}\mid u\in N[v]\}$}]{};
            \node [label={above:$=$}] at (-1.8,-0.5){};

			\node[draw,label={above:$v_i$}] (v2) at (0,0) {};
			\node[draw] (v11) at (-1,0.75) {};
			\node[draw] (v12) at (-1,-0.75) {};
			\node[draw] (v31) at (1,0.75) {};
			\node[draw] (v32) at (1,-0.75) {};
   
            \node at (1,0.15) {\vdots};
            \node at (-1,0.15) {\vdots};
            \draw (1,0) ellipse (15pt and 35pt);
            \draw (-1,0) ellipse (15pt and 35pt);			
            \path (v31) edge[-] (v2);
			\path (v32) edge[-] (v2);
			\path (v11) edge[-] (v2);
			\path (v12) edge[-] (v2);
        \end{tikzpicture}        
    \caption{\label{fig:outpolyhard_bipartite} Construction for \autoref{thm:outpolyhard_bipartite} ($v\in V, i\in \{2,3\}$).}
    \end{figure}
    Clearly, $G$ is bipartite with the classes $V_1 \cup V_3 \cup V_2' \cup V_4'$ and $V_2 \cup V_4 \cup V_1' \cup V_3'$. Furthermore, the graph has at most $ 4 \vert V\vert^2 + 13 \vert V\vert $ vertices.

    Let $A$ be a globally minimal defensive alliance. As the vertices in $ V'' \coloneqq V_1' \cup V_2' \cup V_3' \cup V_4'$ are pendant, $A\cap V''\neq \emptyset$ would imply $\vert A \vert = 1$. Thus, avoiding trivialities, we can assume   $A\cap V'' = \emptyset$. 

    \begin{claim}
         If $A\cap V'' = \emptyset$, then $V_2 \cup V_3 \cup V_4 \subseteq A$.  
    \end{claim}
    \begin{pfclaim} Recall that $V'\setminus V''=V_1\cup V_2 \cup V_3 \cup V_4$. We first show:\\     
    \underline{(A) $N_{V'\setminus V''}(A\cap (V_3\cup V_4))\subseteq A$.} \\
    Assume there exists a $v_4\in A \cap V_4$, then $\deg_{V''}(v_4) = \deg(v)+2 = \deg_{V' \setminus V''}(v_4) + 1$. Hence, $N_{V'\setminus V''}(v_4)\subseteq A$. Analogously, $v_3 \in A \cap V_3$ implies $N_{V'\setminus V''}(v_3)\subseteq A$ by $\deg_{V''}(v_3) = 2 \deg(v) + 3 = \deg_{V' \setminus V''}(v_3) + 1$. 
\\    
    \underline{(B) If $(V_3 \cup V_4)\cap A\neq \emptyset$, then $V_2 \cup V_3 \cup V_4 = N_{V'\setminus V''}[V_3 \cup V_4] \subseteq A$.}\\
    Assume there exists a $v_4\in A \cap V_4$ (or, equivalently, $v_3\in A\cap V_3$) and consider any $u\in V$. We show by induction on the distance between $u$ and $v$ in~$G$ that then, $u_i\in V_i\cap A$ for $i=2,3,4$. This is trivial when the distance is zero by the construction of~$G'$. Assume the claim is true for all vertices at distance $k$ and consider a vertex~$u$ at distance $k+1$. 
    Since $G$ is connected, there is a shortest path $p=v,\dots,u',u$ from $v$ to $u$ in~$G$. By  the induction hypothesis, $u'_3\in V_3\cap A$ and $u'_4\in V_4\cap A$. As $u\in N_G(u')$, by construction and (A), $u_i\in A$ for $i=2,3,4$.
\\
\underline{(C) $N_{V_2}(V_1\cap A)\subseteq A$.}\\ 
Consider some $v_1\in A \cap V_1$. Observing $\deg_{V''}(v_1) = \deg(v) + 2 = \deg_{V' \setminus V''}(v_1) + 1$ implies  $N_{V_2}(v_1) \subseteq A$.
\\
\underline{(D) If $A\cap V_2\neq\emptyset$, then $A\cap A_3\neq \emptyset$.}\\ 
    Assume there is a $v_2\in A \cap V_2$ and $A\cap (V'' \cup V_3 \cup V_4)=\emptyset$. Then $\deg_{V''\cup V_3}(v_2) = \deg(v) + 3 > \deg(v) + 2 = \deg_{V'\setminus (V' \cup V_3)}(v_2)$. Hence, if $A \cap V'' = \emptyset$, then $V_3 \cap A\neq \emptyset$. 

Now, we are ready to prove the claim. We assume that $A\cap V'' = \emptyset$. Hence, some alliance vertex~$x$ lies in $V'\setminus V''=V_1\cup V_2 \cup V_3 \cup V_4$, as minimal alliances are non-empty.
If $x\in V_1$, then by (C) and as $G$ is connected, $A\cap V_2\neq \emptyset$. If $x\in V_2$, then by (D), $A\cap V_3\neq \emptyset$. If $x\in V_3\cup V_4$, $V_2 \cup V_3 \cup V_4 \subseteq A$ follows by (B).  
    \end{pfclaim}

By this claim, any defensive alliance $A\subseteq V_1\cup V_2\cup V_3\cup V_4$ is determined by the the set $A\cap V_1$. We link this to the dominating sets of~$G$ in the following.
    For $D \subseteq V$, define $A_D:=\{v_1\mid v\in D \} \cup V_2 \cup V_3 \cup V_4$. There is a trivial bijection between $D$ and $A_D\cap V_1$.

    \begin{claim}
        Let $D \subseteq V$. $D$ is a dominating set of~$G$ if and only if $A_D$ is a defensive alliance of~$G'$.
    \end{claim}
    \begin{pfclaim}
        Assume $D$ is not a  dominating set. Then there exists a $v_2\in A_D$ such that $N_{A_D}(v_2) \cap V_1 = \emptyset$ and $\deg_{A_D}(v_2) + 1 = \deg(v)+2< \deg(v) + 3 =\deg_{V'\setminus A_D}(v_2)$. This would imply that $A_D$ is not a defensive alliance. 

        Assume $D$ is a dominating set. For $v\in V$, $\deg_{A_D}(v_1) + 1= \deg(v)+2 = \deg_{V' \setminus A_D}(v_1)$, $\deg_{A_D}(v_3) + 1= 2$,  $\deg(v) + 3 = \deg_{V'\setminus A_D }(v_3)$ and $\deg_{A_D}(v_4) = \deg(v)+2 = \deg_{V' \setminus A_D}(v_4) + 1$. Thus, we only need to consider $v_2 \in V_2$ for $v\in V$. Since $D$ is a dominating set, there exists a $u_1\in N_{A_D \cap V_1}(v_2)$. Hence $\deg_{A_D}(v_2) + 1 \geq \deg(v) + 3 > \deg(v)+2 \geq \deg_{V' \setminus A_D}(v_2)$. 
    \end{pfclaim}

    By this claim, we know that for each minimal dominating set of~$G$, there exists a globally minimal defensive alliance of~$G'$. Furthermore, we have shown that either, there exists a $p\in V''$ such that $A=\{p\}$ (at most $4 \vert V\vert^2 +9\vert V\vert$ many tests) or, $V_2\cup V_3 \cup V_4\subseteq A \subseteq V' \setminus V''= V_1 \cup V_2\cup V_3 \cup V_4$. In the last case, there exists a $D\subseteq V$ such that $A=A_D$. \qed  
\end{proof}

It should be mentioned that this transformation does not bring any results for enumerating locally minimal defensive alliances. The problem is that there are locally minimal defensive alliances $A\subseteq V'$ with $A\cap V''\neq \emptyset$ and $A \nsubseteq V''$. For example, let $I \subseteq V$ be, an irredundant set (i.e., each vertex in $I$ has a private neighbor) of~$G$. Then $A_I':=\{v_1\mid v\in V\} \cup \{v_{2,1}\mid v\in V\setminus N[I]\} \cup V_2 \cup V_3 \cup V_4$. Analogously to the proof above, $A_I'$ is a defensive alliance.

Assume $A_I'$ is not locally minimal. As for each $v_i \in \{v_1\mid v\in V\} \cup V_3 \cup V_4$, $\deg_{A_I'}(v_i)+1=\deg_{V'\setminus A_I'}(v_i)$, we cannot delete any vertex from $V_2 \cup V_3 \cup V_4$ from $A_I'$ without losing the defensive alliance property. Let $v_1\in A_I'\cap V_1$, so $v\in I$. As $v$ has a private neighbor $u\in N[u] \setminus N[I\setminus \{v\}]$, $\deg_{A_I' \setminus \{v_1\}}(u_2)+1 = \deg(u) + 2 < \deg(u) + 3 = \deg_{V' \setminus (A_I' \setminus \{v_1\})}(u_2)+1$. This would contradict the defensive alliance property. Analogously, $A_I' \setminus \{v_{2,1}\}$ is  no defensive alliance for $v_{2,1}\in\{v_{2,1}\mid v\in V\setminus N[I]\}$. Hence, $A_I'$ is a locally minimal defensive alliance. There are even more locally minimal defensive alliances.

We conclude this section by adding a further (similar) reduction concerning enumeration of globally minimal
defensive alliances, now on split graphs instead of bipartite graphs.

\begin{theorem}\label{thm:outpolyhard_split}
If there would be an output-polynomial time algorithm for enumerating globally minimal
defensive alliances on split graphs, then there would be an output-polynomial time algorithm for enumerating minimal dominating sets in general graphs. 
\end{theorem}

\begin{pf}
    Let $G=(V,E)$ be a graph. Define $V_i:=\{ v_i \mid v \in V\}$ for $i\in [7]$, $\ell = 6 \cdot \vert V \vert+4$ and $M\coloneqq\{ m_{x,i} \mid x \in V_3 \cup V_4 \cup \{b,c\}, i\in [\ell]\}$ as well as $G'=(V',E')$ with $V'\coloneqq C \cup I$, where
    \begin{equation*}
        \begin{split}
            I \coloneqq{}& V_1  \cup V_6 \cup V_7 \cup M,  \\
            C \coloneqq{}& V_2  \cup V_3\cup V_4 \cup V_5 \cup \{a, b, c, d\}, \\
            E' \coloneqq{}& \binom{C}{2} \cup\{\{v_6, u_2\},\{v_6, u_3\}\mid \{v,u\}\in E\} \cup{} \\
            &\{\{v_6,a\},\{v_6,b\},\{v_7,a\} \mid v\in V\}  \cup\{\{v_1,u_2\},\{v_1,u_3\}\mid v,u\in V, u\in N[v]\}  \cup{}\\ & \{\{m_{x,i},x\} \mid x\in  V_3 \cup V_4 \cup \{b,c\} , i\in [\ell] \}. \\
        \end{split}
    \end{equation*}
    
     \begin{figure}[bt]
    \centering
    	
	\begin{tikzpicture}[transform shape]
		      \tikzset{every node/.style={ fill = white,circle,minimum size=0.3cm}}
            \node  at (4.7,0)[label={[rotate=90]above:$\{u_{6}\mid \{v,u\} \in E \}$}]{};
            \node [label={above:$=$}] at (2.8,-0.5){};
            \node  at (-2,0)[label={[rotate=90]above:$\{u_{1}\mid u\in N[v]\}$}]{};
            \node [label={above:$=$}] at (-2.8,-0.5){};

			\node[draw,label={above:$v_2$}] (v2) at (0,0.5) {};
			\node[draw,label={above:$v_3$}] (v3) at (0,-0.5) {};
			\node[draw,label={above:$v_7$}] (v7) at (-1.5,1.5) {};
			\node[draw,label={above:$a$}] (a) at (0,1.5) {};
			\node[draw,label={above:$b$}] (b) at (0,-1.5) {};
			\node[draw] (v11) at (-2,0.75) {};
			\node[draw] (v12) at (-2,-0.75) {};
			\node[draw] (v61) at (2,0.75) {};
			\node[draw] (v62) at (2,-0.75) {};
   
            \node at (2,0.15) {\vdots};
            \node at (-2,0.15) {\vdots};
            \draw (2,0) ellipse (15pt and 35pt);
            \draw (-2,0) ellipse (15pt and 35pt);			
            \path (v61) edge[-] (v2);
			\path (v62) edge[-] (v2);
			\path (v11) edge[-] (v2);
			\path (v12) edge[-] (v2);	
   
            \path (v61) edge[-] (v3);
			\path (v62) edge[-] (v3);
			\path (v11) edge[-] (v3);
			\path (v12) edge[-] (v3);
   			
            \path (v61) edge[-] (b);
			\path (v62) edge[-] (b);
			\path (v61) edge[-] (a);
			\path (v62) edge[-] (a);
			\path (v7) edge[-] (a);
        \end{tikzpicture}        
    \caption{\label{fig:outpolyhard_split} Construction for \autoref{thm:outpolyhard_split} ($v\in V$). We did not include the edges in $\binom{C}{2}$ and $M=\{\{m_{x,i},x\} \mid x\in  V_3 \cup V_4 \cup \{b,c\} , i\in [\ell] \}$.}
    \end{figure}
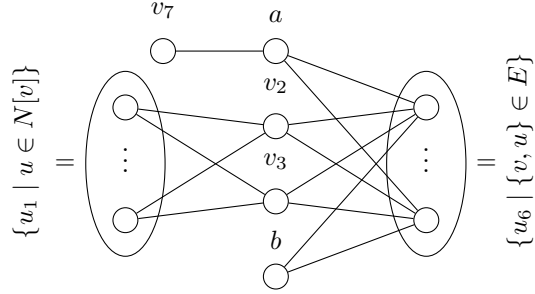
    Clearly, as $C$ is a clique and $I$ is an independent set, $G'$ is a split graph.
    
    Let $A\subseteq V'$ (with $A\neq \emptyset$) be a defensive alliance. As the vertices in $V_7 \cup M$ are pendant, $A\cap (V_7 \cup M)\neq \emptyset$ would imply $\vert A \vert = 1$. Hence, we assume $A\cap (V_7 \cup M) = \emptyset$. Since $\deg_{V_7 \cup M}(x) = 6 \vert V \vert + 4 > 4 \vert V\vert + 2\deg(v) + 4 \geq \deg_{\overline{V_5 \cup M}}(x) + 1$ for $x\in V_3\cup V_4\cup \{b,c\}$, $A\cap (V_3 \cup V_4 \cup \{ b, c \}) =\emptyset$. Define $W:= V_3 \cup V_4 \cup V_7 \cup \{ b, c \} \cup M$. We have seen that for any defensive alliance~$A$ with at least two vertices, $A\subseteq V'\setminus W= V_1\cup V_2 \cup V_5 \cup V_6 \cup \{a,d\}.$ 
    We will follow a similar proof structure as in the previous proof to show that our reduction works.

    \begin{claim}
        If $A\subseteq V'\setminus W$, then $V_2\cup V_5 \cup V_6 \cup \{a,d\}\subseteq A$. 
    \end{claim}

\begin{pfclaim} As $A\subseteq V'\setminus W= V_1\cup V_2 \cup V_5 \cup V_6 \cup \{a,d\}$, we have to discuss cases again, since $A\neq\emptyset$.
\\
\underline{(A) $N_{V_2\cup\{a\}}((V_1\cup V_6)\cap A)\subseteq A$.}
\\
    Assume there is a $v_1\in V_1 \cap A$. Then $\deg_{W}(v_1)=\deg(v)+1$ and $\deg_{\overline{W}}(v_1) + 1=\deg(v)+2$. This implies $N_{V_2}(v_1) \subseteq A$. Analogously, $v_6 \in V_6 \cap A$ implies $N_{V_2}(v_6) \cup \{ a \} \subseteq A$. 
\\
\underline{(B) $A\cap V_2\neq \emptyset$ implies $A\cap (\{a,d\} \cup V_5)\neq \emptyset$.}\\
    Assume $A \cap (V_5 \cup \{a,d\})= \emptyset$ but there is a $v_2\in A\cap V_2$. Then $\deg_{\overline{A}}(v_2) \geq 3 \vert V\vert + 4 > 2\deg(v) +\vert V\vert + 1 \geq \deg_A(v_2)$. This contradicts the defensive alliance property. Thus, $A\cap V_2\neq \emptyset$ implies $A\cap (\{a,d\} \cup V_5)\neq \emptyset$. 
    \\
\underline{(C) $A\cap (\{a,d\} \cup V_5)\neq \emptyset$ implies $V_2 \cup \{ a \} \subseteq A$.}\\
    If there is a $y\in (\{d\} \cup V_5) \cap A$, then $\deg_W(y) = 2 \vert V\vert + 2$ and $\deg_{\overline{W}}(y) + 1 = 2 \vert V\vert + 2$ imply $V_2 \cup \{ a \} \subseteq A$.
        \\
\underline{(D) $a\in A$ implies $V_2\cup V_5 \cup V_6 \cup \{d\} \subseteq A$.}\\If $a\in A$ and $V_2\cup V_5 \cup V_6 \cup \{d\} \nsubseteq A$, then $\deg_{\overline{A}}(a) \geq 3 \vert V\vert + 3 > 3 \vert V\vert + 2\geq \deg_A(a)$ contradicts the defensive alliance property.

Consider $A\subseteq V_1\cup V_2 \cup V_5 \cup V_6 \cup \{a,d\}$.
By (D), if $a\in A$, the claim $V_2\cup V_5 \cup V_6 \cup \{a,d\}\subseteq A$ immediately follows. If  $A\cap (\{d\} \cup V_5)\neq \emptyset$, then (C) shows $a\in A$ and hence the claim. If $A\cap V_2\neq \emptyset$, then (B) proves that the preconditions of (C) are satisfied and thus the claim is true. Finally, if $A\cap (V_1\cup V_6)\neq \emptyset$, then (A) shows that the preconditions of (B) are satisfied, so that again  the claim is true.
Hence, $A\setminus W \neq \emptyset$ implies $V_2\cup V_5 \cup V_6 \cup \{a,d\}\subseteq A$.   
\end{pfclaim}

    So again, any non-trivial alliance~$A$ is determined by the set $A\cap V_1$.  
     For $D \subseteq V$, define $A_D:= \{v_1\mid v\in D\} \cup V_2 \cup V_5 \cup V_6 \cup \{a,d\}.$  
    \begin{claim}
        Let $D\subseteq V$. $D$ is a dominating set of~$G$ if and only if $A_D$ is a defensive alliance of~$G'$.
    \end{claim}
    \begin{pfclaim}
        Let $D$ be a dominating set. Then by the argumentation from above we only need to consider $v_2\in V_2$. As $D$ is a dominating set, there is a $u_1\in N_{V_1 \cap A_D}(v_2)$. Hence, $\deg_{A_D}(v_2) + 1 \leq 2\vert V\vert + \deg(v) + 3 > \deg_{\overline{A_D}}(v_2)$. Thus, $A_D$ is a defensive alliance. 

        Assume $D$ is no dominating set. Then there exists a $v\in V$ such that $N[v]\cap D = \emptyset$. Hence, $\deg_{A_D}(v_2) + 1 = 2 \vert V \vert + \deg(v) + 2 = 2 \vert V \vert + \deg(v) + 3$ implies that $A_D$ is no defensive alliance.        
    \end{pfclaim}
        By this claim, we know that for each minimal dominating set of~$G$ there exists a globally minimal defensive alliance of~$G'$. Furthermore, we have shown that either there exists a $p\in V_7 \cup M$ such that $A=\{p\}$ (at most $12 \vert V\vert^2 +21\vert V\vert +8$ many tests) or $V_2\cup V_5 \cup V_6 \cup \{a,d\} \subseteq A \subseteq V' \setminus W = V_1 \cup V_2\cup V_5 \cup V_6 \cup \{a,d\}$. In the last case, there exists a $D\subseteq V$ such that $A=A_D$. 
\end{pf}

\section{Enumerating All Defensive Alliances With Polynomial Delay}
\label{sec:enum-all-DA}
As we have seen, the enumeration of all minimal defensive alliances is a difficult task in several ways. We also know that in general, we can expect many locally and also globally minimal defensive alliances in a graph, even as many as $\mathcal{O}(2^n)$ in an $n$-vertex graph. As being a defensive alliance is a non-monotone property, it is even not clear if all defensive alliances (again, we can expect $\mathcal{O}(2^n)$ many) can be enumerated with polynomial delay. Observe that this task becomes trivial in the case of dominating sets, for instance, because of monotonicity.
However, we can show that the trivial branching algorithm that decides for each vertex whether it belongs to a defensive alliance or not can be supplemented by a test procedure that guarantees polynomial delay.

\begin{lemma}
There is a polynomial-time algorithm $\textsf{Test-DA}$ that, given a graph $G=(V,E)$ and two disjoint vertex sets $U,N$, answers \yes if and only if there exists a defensive alliance $A\subseteq V$ such that $U\subseteq A$ and $A\cap N=\emptyset$. 
\end{lemma}

\begin{pf}The algorithm $\textsf{Test-DA}$ works as follows:
It first sets $X\coloneqq V\setminus N$. Clearly, $U\subseteq X$.
If $X$ is a defensive alliance, then it answers \yes.
Otherwise, if there exists some $x\in X\setminus U$ that violates the defensive alliance property, i.e., $\deg_{V\setminus X}(x)>\deg_X(x)+1$, then we recursively call $\textsf{Test-DA}(G,U,N\cup\{x\})$; else (i.e., all $x\in X$ that violate the defensive alliance property belong to $U$, which is true in particular if $X=U$) return \no.

One can prove that the procedure will find not just any alliance $A\subseteq V$ such that $U\subseteq A$ and $A\cap N=\emptyset$, but this will be the largest possible, and, moreover, it is uniquely determined, as it is independent on the choice of the vertex   $x\in X\setminus U$ in the course of the algorithm.
Namely, if there are two vertices $a,b\in X\setminus U$ that violate the defensive alliance property, then $b$ will still violate this property after moving $a$ into~$N$ in the recursive call. This is because
$$\deg_{\overline{X\setminus \{a\}}}(b)\geq{} \deg_{\overline{X}}(b)>\deg_X(b)+1\geq \deg_{X\setminus \{a\}}+1
$$
As the defensive alliance property can be checked in polynomial time, the claim follows.
\end{pf}

This lemma entails, by using the mentioned trivial branching algorithm that will build up the sets $U$ and $N$ needed as arguments of the procedure $\textsf{Test-DA}$, the following result.
\begin{theorem}
There is an enumeration algorithm that lists all defensive alliances of a given graph~$G$ with polynomial delay and in polynomial space.
\end{theorem}

\section{Conclusions}
\label{sec:conclusions}

We have shown enumeration algorithms for two variations of minimal defensive alliances that are simple yet optimal from an input-sensitive perspective. The \NP-hardness of the corresponding extension problem (also in quite restrictive settings, like degree-bounded planar graphs) proves that we need other techniques to get output-polynomial enumeration algorithms. This is clearly not ruled out so far and is an interesting open question. In the case of the global problem variation, we prove a link to the famous \textsc{Hitting Set Transversal Problem}. If such a link exists for the local problem variation is again open. 
A more modest goal would be to look for polynomial-time algorithms for the extension problem in graph classes similar as described in \autoref{thm:Enum_global_general}. 

Alliance theory offers quite a number of variants of alliances, for instance, offensive or dual (powerful) or global alliances. There also exist relaxations on the very definition of defensive alliances.
For none of these variations, enumeration has ever been studied. This leaves ample room for further research.
    
\end{toappendix}

\paragraph{Acknowledgements.} We are grateful for the support of Zhidan Feng through China Scholarship Council. 

\bibliographystyle{splncs04}
\bibliography{ab,hen,ref}

\begin{thebibliography}{10}
\providecommand{\url}[1]{\texttt{#1}}
\providecommand{\urlprefix}{URL }
\providecommand{\doi}[1]{https://doi.org/#1}

\bibitem{AbuFGLM2022a}
Abu-Khzam, F.N., Fernau, H., Gras, B., Liedloff, M., Mann, K.: Enumerating
  minimal connected dominating sets. In: Chechik, S., Navarro, G., Rotenberg,
  E., Herman, G. (eds.) 30th Annual European Symposium on Algorithms, ESA.
  Leibniz International Proceedings in Informatics (LIPIcs), vol.~244, pp.
  1:1--1:15. Schloss Dagstuhl -- Leibniz-Zentrum f{\"u}r Informatik (2022)

\bibitem{AbuHeg2016}
Abu{-}Khzam, F.N., Heggernes, P.: Enumerating minimal dominating sets in
  chordal graphs. Information Processing Letters  \textbf{116}(12),  739--743
  (2016)

\bibitem{BazFerTuz2019}
Bazgan, C., Fernau, H., Tuza, Z.: Aspects of upper defensive alliances.
  Discrete Applied Mathematics  \textbf{266},  111--120 (2019)

\bibitem{BlaFLMS2022}
Bl{\"{a}}sius, T., Friedrich, T., Lischeid, J., Meeks, K., Schirneck, M.:
  Efficiently enumerating hitting sets of hypergraphs arising in data
  profiling. Journal of Computer and System Sciences  \textbf{124},  192--213
  (2022)

\bibitem{BlaFriSch2022}
Bl{\"{a}}sius, T., Friedrich, T., Schirneck, M.: The complexity of dependency
  detection and discovery in relational databases. Theoretical Computer Science
   \textbf{900},  79--96 (2022)

\bibitem{BliWol2018}
Bliem, B., Woltran, S.: Defensive alliances in graphs of bounded treewidth.
  Discrete Applied Mathematics  \textbf{251},  334--339 (2018)

\bibitem{BorEGK2002}
Boros, E., Elbassioni, K.M., Gurvich, V., Khachiyan, L.: Generating
  dual-bounded hypergraphs. Optimization Methods and Software  \textbf{17}(5),
  749--781 (2002)

\bibitem{BorGKM2000}
Boros, E., Gurvich, V., Khachiyan, L., Makino, K.: Dual-bounded generating
  problems: Partial and multiple transversals of a hypergraph. {SIAM} Journal
  on Computing  \textbf{30}(6),  2036--2050 (2000)

\bibitem{CasFGMS2022}
Casel, K., Fernau, H., Ghadikolaei, M.K., Monnot, J., Sikora, F.: On the
  complexity of solution extension of optimization problems. Theoretical
  Computer Science  \textbf{904},  48--65 (2022)

\bibitem{ConGMUV2022}
Conte, A., Grossi, R., Marino, A., Uno, T., Versari, L.: Proximity search for
  maximal subgraph enumeration. {SIAM} Journal on Computing  \textbf{51}(5),
  1580--1625 (2022)

\bibitem{DomMisPit99}
Domingo, C., Mishra, N., Pitt, L.: Efficient read-restricted monotone {CNF/DNF}
  dualization by learning with membership queries. Mach. Learn.
  \textbf{37}(1),  89--110 (1999)

\bibitem{EitGot95}
Eiter, T., Gottlob, G.: Identifying the minimal transversals of a hypergraph
  and related problems. {SIAM} Journal on Computing  \textbf{24}(6),
  1278--1304 (1995)

\bibitem{EitGotMak2003}
Eiter, T., Gottlob, G., Makino, K.: New results on monotone dualization and
  generating hypergraph transversals. {SIAM} Journal on Computing
  \textbf{32}(2),  514--537 (2003)

\bibitem{FerRai07}
Fernau, H., Raible, D.: Alliances in graphs: a complexity-theoretic study. In:
  Leeuwen, J., Italiano, G.F., Hoek, W., Meinel, C., Sack, H., Pl\'a\v{s}il,
  F., Bielikov\'a, M. (eds.) SOFSEM 2007, Proceedings Vol. II. pp. 61--70.
  Institute of Computer Science ASCR, Prague (2007)

\bibitem{FerRod2014a}
Fernau, H., Rodr\'{\i}guez-Vel\'{a}zquez, J.A.: A survey on alliances and
  related parameters in graphs. Electronic Journal of Graph Theory and
  Applications  \textbf{2}(1),  70--86 (2014)

\bibitem{FriLHHH2003}
Fricke, G., Lawson, L., Haynes, T.W., Hedetniemi, S.M., Hedetniemi, S.T.: A
  note on defensive alliances in graphs. Bulletin of the Institute of
  Combinatorics and its Applications  \textbf{38},  37--41 (2003)

\bibitem{GaiMai2022}
Gaikwad, A., Maity, S.: Defensive alliances in graphs. Theoretical Computer
  Science  \textbf{928},  136--150 (2022)

\bibitem{GaiMai2022a}
Gaikwad, A., Maity, S.: Globally minimal defensive alliances. Information
  Processing Letters  \textbf{177},  106253 (2022)

\bibitem{GaiMaiTri2021}
Gaikwad, A., Maity, S., Tripathi, S.K.: Parameterized complexity of defensive
  and offensive alliances in graphs. In: Goswami, D., Hoang, T.A. (eds.)
  Distributed Computing and Internet Technology - 17th International
  Conference, {ICDCIT}. LNCS, vol. 12582, pp. 175--187. Springer (2021)

\bibitem{GaiMaiTri2021b}
Gaikwad, A., Maity, S., Tripathi, S.K.: Parameterized complexity of locally
  minimal defensive alliances. In: Mudgal, A., Subramanian, C.R. (eds.)
  Algorithms and Discrete Applied Mathematics - 7th International Conference,
  {CALDAM}. LNCS, vol. 12601, pp. 135--148. Springer (2021)

\bibitem{Gan2015}
Ganian, R.: Improving vertex cover as a graph parameter. Discrete Mathematics
  {\&} Theoretical Computer Sience  \textbf{17}(2),  77--100 (2015)

\bibitem{GogPapSid98}
Gogic, G., Papadimitriou, C.H., Sideri, M.: Incremental recompilation of
  knowledge. Journal of Artificial Intelligence Research  \textbf{8},  23--37
  (1998)

\bibitem{GolHKS2020}
Golovach, P.A., Heggernes, P., Kratsch, D., Saei, R.: Enumeration of minimal
  connected dominating sets for chordal graphs. Discrete Applied Mathematics
  \textbf{278},  3--11 (2020)

\bibitem{HayHedHen2021}
Haynes, T.W., Hedetniemi, S.T., Henning, M.A.: Structures of Domination in
  Graphs, Developments in Mathematics, vol.~66. Springer (2021)

\bibitem{JamHedMcC2009}
Jamieson, L.H., Hedetniemi, S.T., McRae, A.A.: The algorithmic complexity of
  alliances in graphs. Journal of Combinatorial Mathematics and Combinatorial
  Computing  \textbf{68},  137--150 (2009)

\bibitem{KanLMN2014}
Kant{\'e}, M.M., Limouzy, V., Mary, A., Nourine, L.: On the enumeration of
  minimal dominating sets and related notions. {SIAM} Journal of Discrete
  Mathematics  \textbf{28}(4),  1916--1929 (2014)

\bibitem{KanUno2017}
Kant{\'{e}}, M.M., Uno, T.: Counting minimal dominating sets. In: Gopal, T.V.,
  J{\"{a}}ger, G., Steila, S. (eds.) Theory and Applications of Models of
  Computation - 14th Annual Conference, {TAMC}. LNCS, vol. 10185, pp. 333--347
  (2017)

\bibitem{KavPapSid93}
Kavvadias, D.J., Papadimitriou, C.H., Sideri, M.: On {H}orn envelopes and
  hypergraph transversals. In: Ng, K., Raghavan, P., Balasubramanian, N.V.,
  Chin, F.Y.L. (eds.) Algorithms and Computation, 4th International Symposium,
  {ISAAC}. LNCS, vol.~762, pp. 399--405. Springer (1993)

\bibitem{Kimetal2005}
Kim, B.J., Liu, J., Um, J., Lee, S.I.: Instability of defensive alliances in
  the predator-\linebreak[3]prey model on complex networks. Physical Reviews E
  \textbf{72}(4),  041906 (2005)

\bibitem{KiyOta2017}
Kiyomi, M., Otachi, Y.: Alliances in graphs of bounded clique-width. Discrete
  Applied Mathematics  \textbf{223},  91--97 (2017)

\bibitem{KobKMO2024}
Kobayashi, Y., Kurita, K., Matsuo, Y., Ono, H.: Enumerating minimal vertex
  covers and dominating sets with capacity and/or connectivity constraints. In:
  Rescigno, A.A., Vaccaro, U. (eds.) Combinatorial Algorithms (Proceeding 35th
  International Workshop on Combinatorial Algorithms IWOCA). LNCS, vol.~??,
  p.~?? Springer (2024)

\bibitem{KobKurWas2022a}
Kobayashi, Y., Kurita, K., Wasa, K.: Linear-delay enumeration for minimal
  {S}teiner problems. In: Libkin, L., Barcel{\'{o}}, P. (eds.) {PODS} '22:
  International Conference on Management of Data. pp. 301--313. {ACM} (2022)

\bibitem{KobKurWas2022}
Kobayashi, Y., Kurita, K., Wasa, K.: Polynomial-delay and polynomial-space
  enumeration of large maximal matchings. In: Bekos, M.A., Kaufmann, M. (eds.)
  Graph-Theoretic Concepts in Computer Science - 48th International Workshop,
  {WG}. LNCS, vol. 13453, pp. 342--355. Springer (2022)

\bibitem{Kou2013}
Kouteck{\'{y}}, M.: Solving hard problems on Neighborhood Diversity. Master
  thesis, Charles University in Prague, Chech Republic, Facultas Mathematica
  Physicaque, Department of Applied Mathematics (Apr 2013)

\bibitem{KriHedHed2004}
Kristiansen, P., Hedetniemi, S.M., Hedetniemi, S.T.: Alliances in graphs.
  Journal of Combinatorial Mathematics and Combinatorial Computing
  \textbf{48},  157--177 (2004)

\bibitem{Lam2012}
Lampis, M.: Algorithmic meta-theorems for restrictions of treewidth.
  Algorithmica  \textbf{64}(1),  19--37 (2012)

\bibitem{Lar93}
Laroche, P.: Satisfiabilité de l-parmi-3 planaire est {NP}-complet. Comptes
  rendus de l'Académie des sciences. Série 1, Mathématique  \textbf{316},
  389--392 (1993)

\bibitem{ManFer2024}
Mann, K., Fernau, H.: Perfect {R}oman domination: Aspects of enumeration and
  parameterization. In: Rescigno, A.A., Vaccaro, U. (eds.) Combinatorial
  Algorithms (Proceeding 35th International Workshop on Combinatorial
  Algorithms IWOCA). LNCS, vol.~??, p.~?? Springer (2024)

\bibitem{ManRai87}
Mannila, H., R{\"{a}}ih{\"{a}}, K.: Dependency inference. In: Stocker, P.M.,
  Kent, W., Hammersley, P. (eds.) VLDB'87, Proceedings of 13th International
  Conference on Very Large Data Bases. pp. 155--158. Morgan Kaufmann (1987)

\bibitem{Mar2013a}
Mary, A.: \'Enum\'eration des dominants minimaux d'un graphe. Ph.D. thesis,
  LIMOS, Universit\'e Blaise Pascal, Clermont-Ferrand, France (Nov 2013)

\bibitem{MooRob2001}
Moore, C., Robson, J.M.: Hard tiling problems with simple tiles. Discrete \&
  Computational Geometry  \textbf{26}(4),  573--590 (2001)

\bibitem{OuaSliTar2018}
Ouazine, K., Slimani, H., Tari, A.: Alliances in graphs: Parameters, properties
  and applications---a survey. AKCE International Journal of Graphs and
  Combinatorics  \textbf{15}(2),  115--154 (2018)

\bibitem{Rei87}
Reiter, R.: A theory of diagnosis from first principles. Artificial
  Intelligence  \textbf{32},  57--95 (1987)

\bibitem{SebLagKhe2012}
Seba, H., Lagraa, S., Kheddouci, H.: Alliance-based clustering scheme for group
  key management in mobile ad hoc networks. The Journal of Supercomputing
  \textbf{61}(3),  481--501 (2012)

\bibitem{Sha2004}
Shafique, K.H.: Partitioning a graph in alliances and its application to data
  clustering. Phd thesis, School of Computer Science, University of Central
  Florida, Orlando, USA (2004)

\bibitem{SzaCza2001}
Szab\'o, G., Cz\'ar\'an, T.: Defensive alliances in spatial models of cyclical
  population interactions. Physical Review E  \textbf{64}(4),  042902 (Sep
  2001)

\bibitem{Was2016}
Wasa, K.: Enumeration of enumeration algorithms. Tech. Rep. 1605.05102, ArXiv,
  Cornell University (2016)

\bibitem{Wot2001}
Wotawa, F.: A variant of {R}eiter's hitting set algorithm. Information
  Processing Letters  \textbf{79},  45--51 (2001)

\bibitem{YerRod2017}
Yero, I.G., Rodr\'iguez-Vel\'azquez, J.A.: A survey on alliances in graphs:
  defensive alliances. Utilitas Mathematica  \textbf{105},  141--172 (2017)

\bibitem{YuLLY2022}
Yu, K., Long, C., Liu, S., Yan, D.: Efficient algorithms for maximal $k$-biplex
  enumeration. In: Ives, Z.G., Bonifati, A., Abbadi, A.E. (eds.) {SIGMOD} '22:
  International Conference on Management of Data. pp. 860--873. {ACM} (2022)

\end{thebibliography}

\end{document}